\documentclass[a4paper,10pt]{article}
\usepackage[ascii]{inputenc}
\usepackage{amsmath,amssymb,amsfonts,amsthm}
\usepackage[margin=3.8cm]{geometry}
\usepackage{bbm}
\usepackage{graphicx}
\usepackage[caption=false]{subfig}
\usepackage[pdftex,bookmarks=false,colorlinks=true,linkcolor=blue,
citecolor=blue,filecolor=black,urlcolor=blue]{hyperref}

\providecommand{\ud}{\mathrm{d}}
\providecommand{\abs}[1]{\lvert#1\rvert}
\providecommand{\vect}[1]{{\boldsymbol{#1}}}
\providecommand{\R}{\mathbb{R}}
\providecommand{\C}{\mathbb{C}}

\providecommand{\N}{\mathbb{N}}

\DeclareMathOperator{\tr}{tr}

\newtheorem{definition}{Definition}

\newtheorem{proposition}[definition]{Proposition}

\title{Matrix-valued Quantum \\ Lattice Boltzmann Method}

\author{Christian B. Mendl
\footnote{
Zentrum Mathematik, Technische Universit\"at M\"unchen,
Boltzmannstra{\ss}e 3, 85747 Garching bei M\"unchen,
Germany. Email: \href{mailto:mendl@ma.tum.de}{mendl@ma.tum.de}}}

\date{January 9, 2015}

\begin{document}

\maketitle

\begin{abstract}
We devise a lattice Boltzmann method (LBM) for a matrix-valued quantum Boltzmann equation, with the classical Maxwell distribution replaced by Fermi-Dirac functions. To accommodate the spin density matrix, the distribution functions become $2 \times 2$ matrix-valued. From an analytic perspective, the efficient, commonly used BGK approximation of the collision operator is valid in the present setting. The numerical scheme could leverage the principles of LBM for simulating complex spin systems, with applications to spintronics.


\smallskip
\noindent\textbf{MSC classes:} 65Z05, 76Y05, 82D40, 81-08
\end{abstract}

\section{Introduction}

A matrix-valued quantum Boltzmann equation has recently been derived from the Hubbard model \cite{BoltzmannHubbard2012,DerivationBoltzmann2013,BoltzmannNonintegrable2013}, with $2 \times 2$ distribution functions representing the spin density matrix. It can be regarded as generalization of the scalar quantum Boltzmann equation, which traces back to the work of Nordheim \cite{Nordheim1928}, Peierls \cite{Peierls1929} and Uehling and Uhlenbeck \cite{UehlingUhlenbeck1933}. Until now the numerical simulations of the matrix-valued version have only been performed in one dimension \cite{BoltzmannHubbard2012,BoltzmannNonintegrable2013}. Here our aim is to devise a lattice Boltzmann method (LBM) for two and three dimensions.

In recent years, numerical methods for the scalar quantum Boltzmann equation have become a topic of active research. Approaches include asymptotics-preserving exponential methods \cite{HuLiPareschi2013,PareschiRussoReview2011}, Fourier representation to efficiently evaluate the collision operator combined with BKG penalization \cite{FilbertHuJin2012,HuYing2012}, and kinetic flux vector splitting schemes \cite{YangShi2006,YangHsiehShi2007,ShiHuangYang2007}. Concerning the classical Boltzmann equation, lattice Boltzmann methods are widely used for complex fluid simulations \cite{McNamaraZanetti1988,HigueraSucciBenzi1989,QianHumieresLallemand1992,HeLuoTheory1997,Succi2001}, and are particularly attractive from a computational perspective due to their straightforward parallelization. Notably, LBM schemes have recently been developed for a relativistic Boltzmann equation, either using the D3Q19 model with applications to astrophysics \cite{MendozaPRL2010,MendozaPRD2010} or a hexagonal lattice adapted to the physical lattice structure and effective Hamiltonian of graphene \cite{FDGraphene2013}. Alternatively to the present work, a more accurate but computationally also more demanding approach uses Fourier discretization to evaluatate the collision operator with spectral accuracy \cite{LuMendl2014}.

Here, we devise a LBM scheme using the common D2Q9 and D3Q19 models for the matrix-valued quantum Boltzmann equation. The method holds the promise to simulate spin-dependent electronic transport, and might be valuable in the emerging field of spintronics \cite{Spintronics2001,Khajetoorians2011,SiliconSpintronics2012}. Note that the ``lattice'' in the present study discretizes the spatial dimension, and is not related to the lattice on the quantum level of the Hubbard Hamiltonian in \cite{BoltzmannHubbard2012}. We will study the mathematical model in Sec.~\ref{sec:Model}. The numerical framework including lattice discretization, numerical integration adapted to the lattice and a polynomial approximation of Fermi-Dirac distributions is provided in Sec.~\ref{sec:Framework}.

\section{Mathematical model}
\label{sec:Model}

We investigate the following matrix-valued quantum Boltzmann equation
\begin{equation}
\label{eq:BoltzmannEquation}
\partial_t\,W(\vect{p},\vect{x},t) + \vect{p} \cdot \vect{\nabla}_{\vect{x}} W(\vect{p},\vect{x},t) = \mathcal{C}[W](\vect{p},\vect{x},t).
\end{equation}
It formally agrees with the scalar quantum Boltzmann equation \cite{Nordheim1928,Peierls1929,UehlingUhlenbeck1933}, and additionally takes the physical spin degree of freedom for spin-$\tfrac{1}{2}$ particles like electrons into account. This means that the distribution (``Wigner'') function $W(\vect{p},\vect{x},t)$ is $2 \times 2$ complex-Hermitian-valued, i.e., $W(\vect{p},\vect{x},t) \in \mathcal{M}_2$ with $\mathcal{M}_2 = \left\{A \in \C^{2\times2} \,:\, A^* = A\right\}$. Physically, $W(\vect{p},\vect{x},t)$ can be interpreted as spin density matrix distribution function of particles with momentum $\vect{p} \in \R^d$ and position $\vect{x} \in \R^d$ at time $t$. (In this paper we will consider dimensions $d = 2, 3$).

The collision operator $\mathcal{C}$ models the interaction of the particles by collisions and is likewise $2 \times 2$ matrix-valued. It acts only locally on the momentum variable $\vect{p}$, i.e., $\mathcal{C}$ is a functional on momentum distributions which does not explicitly depend on $\vect{x}$ or $t$. Specifically, the collision operator derived from the Hubbard model for electrons in the limit of weak interactions \cite{BoltzmannHubbard2012,DerivationBoltzmann2013} reads as follows: $\mathcal{C}$ splits into a ``dissipative'' and ``conservative'' part, $\mathcal{C} = \mathcal{C}_{\mathrm{d}} + \mathcal{C}_{\mathrm{c}}$. The conservative part results in local, $\vect{p}$-dependent unitary rotations. Namely, $\mathcal{C}_{\mathrm{c}}$ is a Vlasov-type commutator (locally at given $\vect{x}, t$):
\begin{equation}
\label{eq:CcHubbard}
\mathcal{C}_{\mathrm{c}}[W](\vect{p}) = - i \left[ H_\mathrm{eff}(\vect{p}), W(\vect{p}) \right],
\end{equation}
where the effective Hamiltonian $H_\mathrm{eff}(\vect{p})$ itself depends on $W$. $\mathcal{C}_{\mathrm{c}}$ is denoted ``conservative'' since it does not increase the entropy, as the entropy is invariant under unitary rotations. For the dissipative part $\mathcal{C}_{\mathrm{d}}$, we introduce the notation $\tilde{W} = \mathbbm{1} - W$, $W_i = W(\vect{p}_i,\vect{x},t)$, $\underline{\vect{p}} = \vect{p}_1 + \vect{p}_2 - \vect{p}_3 - \vect{p}_4$, $\underline{\omega} = \omega(\vect{p}_1) + \omega(\vect{p}_2) - \omega(\vect{p}_3) - \omega(\vect{p}_4)$. Here the analytic function $\omega: \R^d \to \R$ is the dispersion relation (energy). $\omega$ is assumed to be non-negative and symmetric: $\omega(\vect{p}) = \omega(\vect{-p})$. In this paper, we only consider $\omega(\vect{p}) = \frac{1}{2} \abs{\vect{p}}^2$, corresponding to fermions in the continuum. The dissipative $\mathcal{C}_{\mathrm{d}}$ models the collision of two particles with momenta $\vect{p}_1$ and $\vect{p}_2$, resulting in two outgoing particles with momenta $\vect{p}_3$ and $\vect{p}_4$. Conservation of momentum requires that $\vect{p}_1 + \vect{p}_2 = \vect{p}_3 + \vect{p}_4$, which is reflected by the $\delta$-function $\delta(\underline{\vect{p}})$ in Eq.~\eqref{eq:CdHubbard} below. Similarly, the $\delta(\underline{\omega})$ term in Eq.~\eqref{eq:CdHubbard} guarantees energy conservation. The operator $\mathcal{C}_{\mathrm{d}}$ reads
\begin{equation}
\label{eq:CdHubbard}
\mathcal{C}_\mathrm{d}[W]_1 = \pi \int_{\R^{3d}} \delta(\underline{\vect{p}}) \delta(\underline{\omega}) \big( \mathcal{A}_{\mathrm{quad}}[W]_{1234} + \mathcal{A}_{\mathrm{tr}}[W]_{1234} \big) \ud\vect{p}_2 \ud\vect{p}_3 \ud\vect{p}_4,
\end{equation}
where the index $1234$ indicates that the matrices $\mathcal{A}_{\mathrm{quad}}[W]$ and $\mathcal{A}_{\mathrm{tr}}[W]$ depend on $\vect{p}_1$, $\vect{p}_2$, $\vect{p}_3$, and $\vect{p}_4$. Explicitly
\begin{equation}
\begin{split}
\label{eq:AW_Hubbard}
\mathcal{A}_{\mathrm{quad}}[W]_{1234} &= -\tilde{W}_1 W_3 \tilde{W}_2 W_4 - W_4 \tilde{W}_2 W_3 \tilde{W}_1 + W_1 \tilde{W}_3 W_2 \tilde{W}_4 + \tilde{W}_4 W_2 \tilde{W}_3 W_1, \\
\mathcal{A}_{\mathrm{tr}}[W]_{1234} &= \big(\tilde{W}_1 W_3 + W_3 \tilde{W}_1\big) \tr[\tilde{W}_2 W_4] - \big(W_1 \tilde{W}_3 + \tilde{W}_3 W_1\big) \tr[W_2 \tilde{W}_4].
\end{split}
\end{equation}
Note that the Wigner matrices do not commute in general, i.e., $W_i W_j \neq W_j W_i$. The trace $\tr[\,\cdot\,]$ appearing in \eqref{eq:AW_Hubbard} averages the spin components. The representation \eqref{eq:AW_Hubbard} emphasizes the similarity to the scalar collision operator \cite{UehlingUhlenbeck1933}, which is recovered when all $W_i$ are proportional to the identity matrix. After expanding the expressions in \eqref{eq:AW_Hubbard}, at most cubit terms remain. We will analyze the Bhatnagar-Gross-Krook (BGK)~\cite{BGK1954,QianHumieresLallemand1992} approximation of the collision operator in section~\ref{sec:BGK} due to its relevance for numerical implementations.

The equilibrium functions $W_{\mathrm{FD}}$ satisfying $\mathcal{C}[W_{\mathrm{FD}}] = 0$ are of Fermi-Dirac type. That is, after a unitary base change,
\begin{equation}
\label{eq:WFermiDirac}
W_\mathrm{FD}(\vect{p}) =
\begin{pmatrix}
\left(\mathrm{e}^{\beta (\omega(\vect{p}-\vect{u}) - \mu_{\uparrow})} + 1\right)^{-1}& 0\\
0 & \left(\mathrm{e}^{\beta (\omega(\vect{p}-\vect{u}) - \mu_{\downarrow})} + 1\right)^{-1}\\
\end{pmatrix},
\end{equation}
where -- in physical terms -- $\vect{u} \in \R^d$ is the average velocity, $\beta = \frac{1}{k_{\mathrm{B}} T}$ the ``inverse temperature'' with $k_{\mathrm{B}}$ the Boltzmann constant and $T$ the temperature, and $\mu_\uparrow$, $\mu_\downarrow \in \R$ are the chemical potentials for the spin occupations.

\subsection{Abstract characterization of the collision operator}
\label{sec:CollisionOpAbstract}

We summarize the abstract properties of $\mathcal{C}$ \cite{BoltzmannHubbard2012}, analogous to \cite{BardosGolseLevermore1991}. First, $\mathcal{C}$ should act locally as mentioned above, i.e., only on the momentum variable $\vect{p}$. Second, $\mathcal{C}$ must be $\mathrm{SU}(2)$-invariant, that is, $\mathcal{C}[U^*\,W\,U] = U^* \mathcal{C}[W] U$ for arbitrary, constant $U \in \mathrm{SU}(2)$. The collision operator must propagate the Fermi property: the Fermi property is satisfied at time $t$ and position $\vect{x}$ if the two eigenvalues of $W(\vect{p},\vect{x},t)$ are in the interval $[0,1]$ for all $\vect{p}$. If the Fermi property holds at some initial time $t = t_0$, it must be preserved at all later times $t \ge t_0$ under the time evolution of the Boltzmann equation. Furthermore, the collision operator $\mathcal{C}$ must adhere to the following density, momentum and energy conservation laws ($\vect{x}, t$ held fixed):
\begin{equation}
\label{eq:ConservationLaws}
\int_{\R^d} \mathcal{C}[W](\vect{p}) \,\ud\vect{p} = 0, \quad
\int_{\R^d} \vect{p} \tr[\mathcal{C}[W](\vect{p})] \,\ud\vect{p} = 0, \quad
\int_{\R^d} \tfrac{1}{2} \abs{\vect{p}}^2 \tr[\mathcal{C}[W](\vect{p})] \,\ud\vect{p} = 0.
\end{equation}
The corresponding fluid dynamic moments, i.e., density $\rho(\vect{x},t)$, velocity $\vect{u}(\vect{x},t)$ and internal energy $\varepsilon(\vect{x},t)$ are locally (at given $\vect{x},t$) defined by the relations
\begin{align}
\label{eq:AvrDensity}
\rho &= \int_{\R^d} W(\vect{p}) \,\ud\vect{p},\\
\label{eq:AvrVelocity}
\tr[\rho]\,\vect{u} &= \int_{\R^d} \vect{p} \tr[W(\vect{p})] \,\ud\vect{p},\\
\label{eq:AvrEnergy}
\tr[\rho]\,\varepsilon &= \int_{\R^d} \tfrac{1}{2} \abs{\vect{p}}^2 \tr[W(\vect{p} + \vect{u})] \,\ud\vect{p} = \int_{\R^d} \tfrac{1}{2} \abs{\vect{p}}^2 \tr[W(\vect{p})] \,\ud\vect{p} - \tr[\rho]\,\tfrac{1}{2} \abs{\vect{u}}^2.
\end{align}
Note that $\rho$ is $2 \times 2$ matrix-valued. The shift by $\vect{u}$ in Eq.~\eqref{eq:AvrEnergy} relocates the integration variable into the moving frame of $W$. The total energy (including the kinetic part resulting from the overall motion with velocity $\vect{u}$) equals $\varepsilon + \frac{1}{2} \abs{\vect{u}}^2$. When additionally defining the stress tensor $R$ and heat flux $\vect{q}$ by
\begin{equation}
R = \int_{\R^d} \vect{p} \otimes \vect{p} \tr[W(\vect{p} + \vect{u})] \,\ud\vect{p}, \qquad \vect{q} = \int_{\R^d} \tfrac{1}{2} \abs{\vect{p}}^2\,\vect{p} \tr[W(\vect{p} + \vect{u})] \,\ud\vect{p},
\end{equation}
then the Boltzmann equation~\eqref{eq:BoltzmannEquation} together with \eqref{eq:ConservationLaws} immediately leads to the following local conservation laws (as in the scalar case \cite{HuLiPareschi2013} after replacing $\rho$ by $\tr[\rho]$)
\begin{equation}
\label{eq:hydrodynamic}
\begin{split}
&\partial_t \tr[\rho] + \vect{\nabla}_{\vect{x}} \cdot (\tr[\rho]\,\vect{u}) = 0, \\[2pt]
&\partial_t (\tr[\rho]\,\vect{u}) + \vect{\nabla}_{\vect{x}} \cdot (R + \tr[\rho]\, \vect{u} \otimes \vect{u}) = 0, \\[2pt]
&\partial_t \big(\tr[\rho] \big(\varepsilon + \tfrac{1}{2} \abs{\vect{u}}^2 \big) \big) + \vect{\nabla}_{\vect{x}} \cdot \big(\vect{q} + R \vect{u} + \tr[\rho]\big(\varepsilon + \tfrac{1}{2} \abs{\vect{u}}^2 \big)\vect{u} \big) = 0.
\end{split}
\end{equation}
Thus these conservation laws are insensitive to the matrix structure of $\rho$.

Finally, the H-theorem should be satisfied, meaning that the entropy cannot decrease under the time evolution of the Boltzmann equation. The entropy of the state $W$ (locally at $\vect{x}, t$) is defined as
\begin{equation}
S[W] = - \int_{\R^d} \tr\!\big[ W(\vect{p}) \log W(\vect{p}) + \tilde{W}(\vect{p}) \log \tilde{W}(\vect{p}) \big] \ud \vect{p},
\end{equation}
where the (natural) logarithm acts on the eigenvalues of its argument. Hence the entropy production is given by
\begin{equation}
\label{eq:sigmaW}
\sigma[W] = \frac{\ud}{\ud t} S[W] = -\int_{\R^d} \tr\!\big[ \big(\log W(\vect{p}) - \log \tilde{W}(\vect{p})\big) \mathcal{C}[W](\vect{p}) \big] \ud \vect{p}.
\end{equation}
The H-theorem asserts that
\begin{equation}
\label{eq:HTheorem}
\sigma[W] \geq 0 \qquad \text{for all } W \text{ with eigenvalues } 0 \leq W \leq 1.
\end{equation}
This property holds indeed for the collision operator \eqref{eq:CdHubbard} derived from the Hubbard model, and we presuppose the same for any abstract collision operator.

\subsection{Fermi-Dirac equilibrium distribution function}

All equilibrium functions satisfying $\mathcal{C}[W] = 0$ are precisely of the form $U\,W_{\mathrm{FD}}(\vect{p})\,U^*$ with $U \in \mathrm{SU}(2)$ (independent of $\vect{p}$) and $W_{\mathrm{FD}}$ defined in Eq.~\eqref{eq:WFermiDirac}. Moreover, the moments suffice to determine $W_{\mathrm{FD}}$, as stated in the following proposition. The proof proceeds along similar lines as in Ref.~\cite{BoltzmannHubbard2012}, using a Legendre transformation.
\begin{proposition}
\label{prop:MapMomentsParams}
The moments of the Fermi-Dirac distribution $W_{\mathrm{FD}}$ defined via Eq.~\eqref{eq:AvrDensity}, \eqref{eq:AvrVelocity} and \eqref{eq:AvrEnergy} uniquely determine the parameters $(\beta, \mu_{\uparrow}, \mu_{\downarrow}, \vect{u})$.
\end{proposition}
\begin{proof}
First note that $\vect{u}$ in Eq.~\eqref{eq:WFermiDirac} coincides with $\vect{u}$ in Eq.~\eqref{eq:AvrVelocity}, which follows by a change of variables $\vect{p} \to \vect{p} + \vect{u}$ and using the even symmetry of the Fermi-Dirac distribution. Thus we can without loss of generality assume that $\vect{u} = 0$ after a shift into the moving frame. In the following, it will be advantageous to work with $\nu_\sigma = \beta \mu_\sigma$ instead of $\mu_\sigma$ ($\sigma \in \{\uparrow, \downarrow\}$). The map between $(\beta, \nu_{\uparrow}, \nu_{\downarrow})$ and the averages $(\rho_{\mathrm{FD}}, \varepsilon_{\mathrm{FD}})$ can be regarded as Legendre transformation: define a ``free energy'' by
\begin{equation}
\label{eq:FreeEnergy}
H(\beta,\nu_{\uparrow},\nu_{\downarrow}) = \int_{\R^d} \sum_{\sigma \in \{\uparrow, \downarrow\}} \log\!\left(1 + \mathrm{e}^{\nu_\sigma - \beta \omega(\vect{p})}\right) \ud\vect{p}.
\end{equation}
The integrand tends exponentially to zero as $\abs{\vect{p}} \to \infty$. Assuming that the order of differentiation and integration can be interchanged, a short calculation of the derivatives of $H$ results in
\begin{equation}
\partial_\beta H = - \int_{\R^d} \omega(\vect{p}) \sum_{\sigma \in \{\uparrow, \downarrow\}} \left(\mathrm{e}^{\beta\omega(\vect{p}) - \nu_{\sigma}} + 1\right)^{-1} \ud\vect{p} = - \tr[\rho_\mathrm{FD}] \, \varepsilon_{\mathrm{FD}}
\end{equation}
according to the definition~\eqref{eq:AvrEnergy}. Also,
\begin{equation}
\partial_{\nu_\sigma} H = \int_{\R^d} \left(\mathrm{e}^{\beta\omega(\vect{p}) - \nu_{\sigma}} + 1\right)^{-1} \ud\vect{p} = \rho_{\mathrm{FD},\sigma},
\end{equation}
where $\rho_{\mathrm{FD},\sigma}$ denotes the diagonal matrix entry of $\rho_{\mathrm{FD}}$ corresponding to $\sigma$. The uniqueness of the map follows from the strict convexity of $H$ in its arguments.
\end{proof}

We remark that for a general equilibrium state $U\,W_{\mathrm{FD}}(\vect{p})\,U^*$, the unitary matrix $U \in \mathrm{SU}(2)$ can be recovered from the average density $\rho \in \mathcal{M}_2$ by diagonalizing $\rho$.

\medskip

Let us briefly discuss the implications of proposition~\ref{prop:MapMomentsParams} for the homogeneous case, i.e., $W(\vect{p},\vect{x},t) = W(\vect{p},t)$ independent of $\vect{x}$. Then the transport term in the Boltzmann equation~\eqref{eq:BoltzmannEquation} disappears, and the moments \eqref{eq:AvrDensity}, \eqref{eq:AvrVelocity}, \eqref{eq:AvrEnergy} are globally conserved. In this case, one can actually calculate the equilibrium distribution which the current state $W(\vect{p},t)$ will eventually converge to, even without solving the Boltzmann time evolution.

\medskip

The average density of the Fermi-Dirac distribution is (see e.g.~\cite{Pathria1997})
\begin{equation}
\label{eq:rhoFDstdOmega}
\rho_{\mathrm{FD}}
= \int_{\R^d} \left(\mathrm{e}^{\beta \left(\frac{1}{2} \abs{\vect{p}-\vect{u}}^2 - \mu_{\sigma}\right)} + 1\right)^{-1}_{\sigma \sigma} \ud\vect{p}
= \begin{pmatrix} n_{d,\beta,\mu_{\uparrow}} & 0 \\ 0 & n_{d,\beta,\mu_{\downarrow}} \end{pmatrix} \qquad \end{equation}
with
\begin{equation}
\label{eq:nFDstdOmega}
n_{d,\beta,\mu} = \left(\frac{2 \pi}{\beta}\right)^{\frac{d}{2}} F_{\frac{d}{2}-1}(\beta \mu),
\end{equation}
and the internal energy
\begin{equation}
\label{eq:energyFDstdOmega}
\varepsilon_{\mathrm{FD}}
= \frac{1}{\tr[\rho_{\mathrm{FD}}]} \int_{\R^d} \tfrac{1}{2}\abs{\vect{p}}^2 \sum_{\sigma \in \{\uparrow, \downarrow\}} \left(\mathrm{e}^{\beta \left(\frac{1}{2} \abs{\vect{p}}^2 - \mu_{\sigma}\right)} + 1\right)^{-1} \ud\vect{p}
= \frac{d}{2} \frac{1}{\beta} \frac{\sum_{\sigma \in \{\uparrow, \downarrow\}} F_{\frac{d}{2}}(\beta \mu_{\sigma})}{\sum_{\sigma \in \{\uparrow, \downarrow\}} F_{\frac{d}{2}-1}(\beta \mu_{\sigma})}.
\end{equation}
Due to the even symmetry of the Fermi-Dirac distribution, $R_{\mathrm{FD}} = \frac{2}{d} \tr[\rho] \varepsilon_{\mathrm{FD}}\,\mathbbm{1}$ and $\vect{q}_{\mathrm{FD}} = 0$. In \eqref{eq:nFDstdOmega} and \eqref{eq:energyFDstdOmega}, $F_{k}(x)$ is the complete Fermi-Dirac integral
\begin{equation*}
F_{k}(x) = \frac{1}{\Gamma(k+1)} \int_0^{\infty} \frac{y^k}{\mathrm{e}^{y - x} + 1} \ud y,
\end{equation*}
which is valid for all $k \in \C$ by analytic continuation to negative integers. For example, $F_{-1}(x) = (1+\mathrm{e}^{-x})^{-1}$. The function $F_{k}(x)$ is related to the polylogarithm function by $F_{k}(x) = -\mathrm{Li}_{k+1}\!\left(-\mathrm{e}^x\right)$ and obeys the following relation: $\frac{\ud}{\ud x} F_{k}(x) = F_{k-1}(x)$. For $k = 0$, one obtains $F_0(x) = \log(1 + \mathrm{e}^x)$. Concerning precise and efficient numerical evaluation of $F_{k}(x)$ for half-integer $k$, see Ref.~\cite{MacLeod1998}.

%and $\mathrm{Li}_k(z)$ in turn admits the series expansion
%%
%\begin{equation}
%\mathrm{Li}_k(z) = \sum_{j=1}^{\infty} \frac{z^j}{j^k}.
%\end{equation}

\medskip

In the limit where $W$ approaches a Fermi-Dirac equilibrium function, we may substitute the above formulas for the stress tensor and heat flux into the system \eqref{eq:hydrodynamic} to obtain the following closed system of equations as hydrodynamic limit, analogous to \cite{HuLiPareschi2013}:
\begin{equation}\label{eq:Euler}
\begin{split}
&\partial_t \tr[\rho] + \vect{\nabla}_{\vect{x}} \cdot (\tr[\rho]\,\vect{u}) = 0, \\[2pt]
&\partial_t (\tr[\rho]\,\vect{u}) + \vect{\nabla}_{\vect{x}} \cdot \big( \tr[\rho] \big( \tfrac{2}{d}\,\varepsilon\,\mathbbm{1} + \vect{u} \otimes \vect{u} \big) \big) = 0, \\[2pt]
&\partial_t \big(\tr[\rho] \big(\varepsilon + \tfrac{1}{2} \abs{\vect{u}}^2\big) \big) + \vect{\nabla}_{\vect{x}} \cdot \big(\tr[\rho]\big(\tfrac{2}{d}\,\varepsilon + \varepsilon + \tfrac{1}{2} \abs{\vect{u}}^2 \big)\vect{u} \big) = 0.
\end{split}
\end{equation}
These are precisely the classical Euler equations (with $\tr[\rho]$ instead of $\rho$) when identifying the pressure as $P = \frac{2}{d} \tr[\rho]\varepsilon$. Comparing with the equation of state for an ideal polytropic gas, the adiabatic exponent is $\gamma = 1 + \frac{2}{d}$, corresponding to a monatomic gas.

\subsection{The BGK collision operator}
\label{sec:BGK}

The Bhatnagar-Gross-Krook (BGK)~\cite{BGK1954} approximation of the collision operator is widely used in LBM schemes \cite{QianHumieresLallemand1992}:
\begin{equation}
\label{eq:C_BGK}
\mathcal{C}_{\mathrm{BGK}}[W](\vect{p}) = \frac{1}{\tau} \big(W^{\mathrm{(eq)}}(\vect{p}) - W(\vect{p}) \big).
\end{equation}
Here $\tau$ is the relaxation time and $W^{\mathrm{(eq)}}$ is the local (at fixed $\vect{x},t$) equilibrium distribution function corresponding to $W$, that is, $W^{\mathrm{(eq)}}$ has the same momentum-averaged density, velocity and energy as $W$ at ($\vect{x},t$). The exponential convergence to equilibrium effected by $\mathcal{C}_{\mathrm{BGK}}$ agrees with the numerical simulations in \cite{BoltzmannHubbard2012}. In our case, the general form of the equilibrium state is
\begin{equation}
\label{eq:WeqFD}
W^{\mathrm{(eq)}}(\vect{p}) = U\,W_{\mathrm{FD}}(\vect{p})\,U^*,
\end{equation}
where the unitary matrix $U \in \mathrm{SU}(2)$ encodes the eigenbasis of the average density such that $U^* \left(\int_{\R^d} W(\vect{p}) \ud\vect{p} \right) U$ is a diagonal matrix. 

We still have to ensure that $\mathcal{C}_{\mathrm{BGK}}$ is indeed valid in the present setting where $W(\vect{p})$ is matrix-valued, i.e., that all properties listed in section~\ref{sec:CollisionOpAbstract} are satisfied. We remark that the following analysis is concerned with the exact solution of the Boltzmann equation using the collision operator \eqref{eq:C_BGK}. Whether all properties (in particular the H-theorem) are reflected within a numerical framework is an issue of its own. The conservation of density, momentum and energy is satisfied by construction, but a more subtle point is the Fermi-property. Our argument proceeds along similar lines as in \cite{BoltzmannHubbard2012}. We denote the eigenvalues of $W(\vect{p})$ at a time point $t$ by $\lambda_\sigma(\vect{p})$, without loss of generality $\lambda_{\uparrow}(\vect{p}) \ge \lambda_{\downarrow}(\vect{p})$. Let us assume that $W$ has an eigenvalue $0$ at $\vect{p}_0$, i.e., $\lambda_{\downarrow}(\vect{p}_0) = 0$ with corresponding eigenvector $\phi_{\downarrow} \in \C^2$. Then by first order perturbation theory,
\begin{equation}
\begin{split}
\frac{\ud}{\ud t} \lambda_{\downarrow}(\vect{p}_0)
= \left\langle \phi_{\downarrow} \,\vert\, \mathcal{C}_{\mathrm{BGK}}[W](\vect{p}_0) \,\vert\, \phi_{\downarrow}\right\rangle
&= \frac{1}{\tau} \big\langle \phi_{\downarrow} \,\vert\, W^{\mathrm{(eq)}}(\vect{p}_0) -  W(\vect{p}_0) \,\vert\, \phi_{\downarrow}\big\rangle \\
&= \frac{1}{\tau} \left\langle \phi_{\downarrow} \,\vert\, U\,W_{\mathrm{FD}}(\vect{p}_0)\,U^* \,\vert\, \phi_{\downarrow}\right\rangle \ge 0.
\end{split}
\end{equation}
The last inequality holds since the Fermi-Dirac distribution $W_{\mathrm{FD}}(\vect{p}_0)$ is a positive-definite matrix with eigenvalues in the interval $(0,1)$. Thus, $\mathcal{C}_{\mathrm{BGK}}$ prevents $\lambda_{\downarrow}(\vect{p}_0)$ from becoming negative. By a similar argument, if $\lambda_{\uparrow}(\vect{p}_0) = 1$, one obtains $\frac{\ud}{\ud t} \lambda_{\uparrow}(\vect{p}_0) \le 0$, that is, the time evolution prevents eigenvalues from exceeding $1$.

\begin{proposition}
The H-theorem, Eq.~\eqref{eq:HTheorem}, holds for the BGK collision operator defined in Eq.~\eqref{eq:C_BGK}. Furthermore, the entropy production $\sigma[W]$ is zero if and only if $W(\vect{p}) = W^{\mathrm{(eq)}}(\vect{p})$.
\end{proposition}
\begin{proof}
Without loss of generality we can assume that the unitary matrix $U$ in Eq.~\eqref{eq:WeqFD} is the identity matrix, since the following arguments are invariant under global unitary rotations. A short calculation shows that
\begin{equation}
\log W^{\mathrm{(eq)}}(\vect{p}) - \log \tilde{W}^{\mathrm{(eq)}}(\vect{p}) = - \begin{pmatrix}\beta (\omega(\vect{p}-\vect{u}) - \mu_{\uparrow}) & 0\\ 0 & \beta (\omega(\vect{p}-\vect{u}) - \mu_{\downarrow})\end{pmatrix}.
\end{equation}
Together with the conservation properties \eqref{eq:ConservationLaws}, it follows that
\begin{equation}
\label{eq:dhWeq}
-\int_{\R^d} \tr\!\big[ \big(\log W^{\mathrm{(eq)}}(\vect{p}) - \log \tilde{W}^{\mathrm{(eq)}}(\vect{p})\big) \mathcal{C}[W](\vect{p}) \big] \ud \vect{p} = 0.
\end{equation}
Define the ``binary entropy function'' $h(x) = - x \log(x) - (1-x)\log(1-x)$ for $x \in (0,1)$, with derivative $h'(x) = -\log(x) + \log(1-x)$. Together with Eq.~\eqref{eq:dhWeq}, the entropy production can then be written as
\begin{equation}
\label{eq:sigmaWeq}
\sigma[W] = \int_{\R^d} \tr\!\big[\big( h'(W(\vect{p})) - h'(W^{\mathrm{(eq)}}(\vect{p}))\big) \mathcal{C}[W](\vect{p}) \big] \ud \vect{p},
\end{equation}
where $h'$ acts on the eigenvalues of its argument. Note that Eq.~\eqref{eq:sigmaWeq} holds for any admissible collision operator $\mathcal{C}$. Specifically for $\mathcal{C} = \mathcal{C}_{\mathrm{BGK}}$, we obtain
\begin{equation}
\label{eq:sigmaW_BGK}
\sigma[W] = \frac{1}{\tau} \int_{\R^d} \tr\!\big[\big( h'(W(\vect{p})) - h'(W^{\mathrm{(eq)}}(\vect{p}))\big) \big(W^{\mathrm{(eq)}}(\vect{p}) - W(\vect{p})\big) \big] \ud \vect{p}.
\end{equation}
To show that the integrand is pointwise (for each $\vect{p}$) non-negative, we use the \emph{quantum relative entropy} defined as $s(A \parallel B) = \tr[A \log A - A \log B]$ for positive-definite matrices $A, B$. By Klein's inequality, $s(A \parallel B) \ge 0$, with equality if and only if $A = B$. A short calculation shows that the integrand can be rewritten in terms of the relative entropy (pointwise in $\vect{p}$) as
\begin{equation}
s\big(W \parallel W^{\mathrm{(eq)}}\big) + s\big(W^{\mathrm{(eq)}} \parallel W\big) + s\big(\tilde{W} \parallel \tilde{W}^{\mathrm{(eq)}}\big) + s\big(\tilde{W} \parallel \tilde{W}^{\mathrm{(eq)}}\big) \ge 0.
\end{equation}
In particular, $\sigma[W] = 0$ can only hold if $W(\vect{p}) = W^{\mathrm{(eq)}}(\vect{p})$ for all $\vect{p}$.
\end{proof}

As a remark, an alternative proof for the non-negativity of the integrand in \eqref{eq:sigmaW_BGK} could rely on the argument that $(h'(x) - h'(y))(y - x) \ge 0$ for all $x, y \in (0,1)$ since $h'$ is a strictly decreasing function. An explicit calculation parameterizing the eigenbasis of $W(\vect{p}))$ could handle the fact that $W(\vect{p}))$ is not diagonal in general.

\section{Numerical framework and discretization}
\label{sec:Framework}

The starting point of LBM is a discretization of the position variable $\vect{x}$ by a uniform grid of cells denoted $\Lambda$, and of the momentum variable $\vect{p}$ by a small number of ``velocity vectors'' pointing from their current grid cell to neighboring grid cells, as illustrated in Fig.~\ref{fig:LBmodels}.
\begin{figure}[!ht]
\centering
\includegraphics[width=0.7\textwidth]{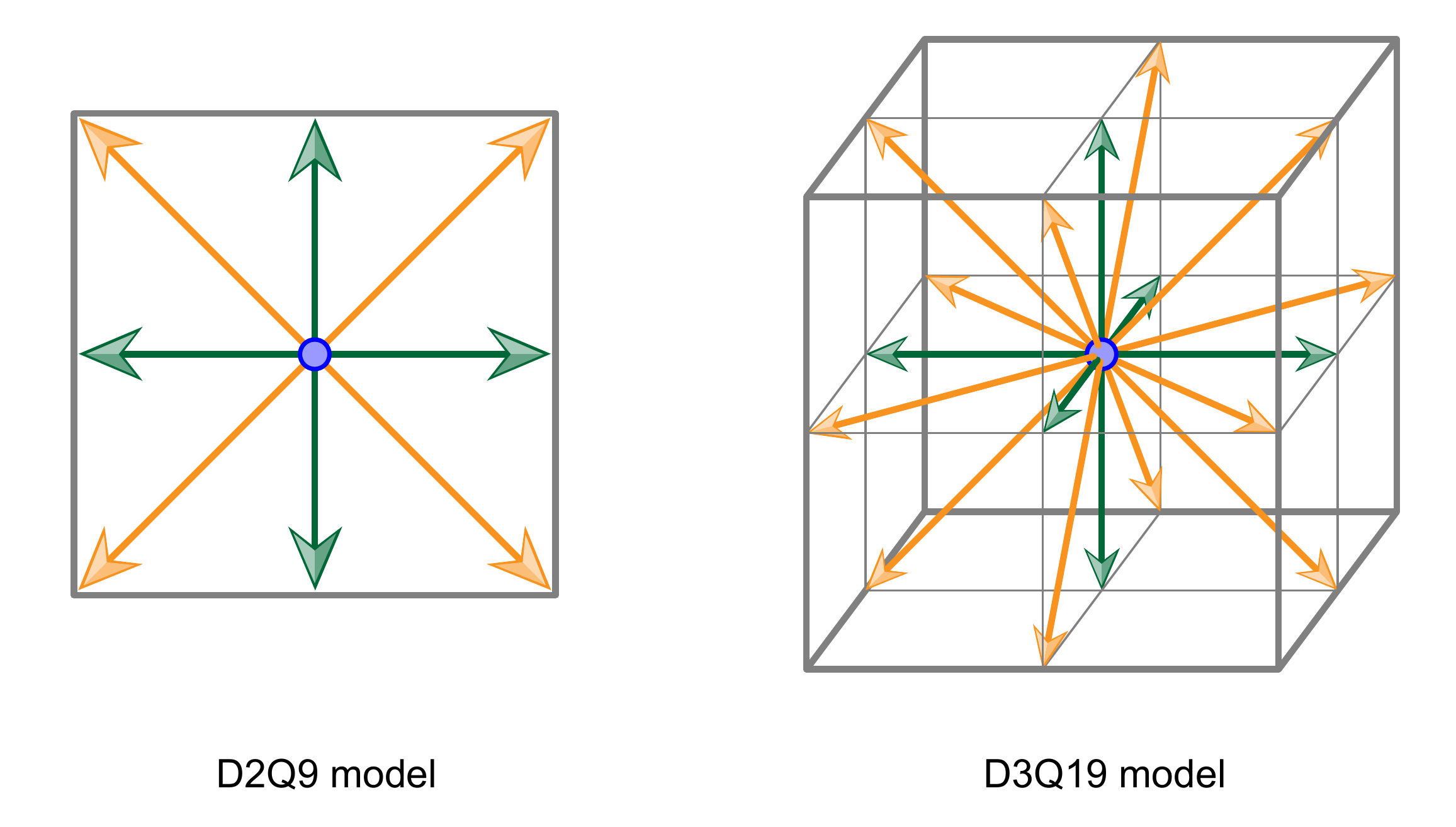}
\caption{(Color online) Illustration of the velocity vectors $\vect{\mathsf{e}}_i$ of the D2Q9 and D3Q19 models (adapted from \cite{ThureyPhD2006}). All arrows and the blue circle are associated with matrix-valued distribution functions $\mathsf{W}_i(\vect{x},t)$ in the present framework.}
\label{fig:LBmodels}
\end{figure}
Here we focus on the most commonly used lattice Boltzmann models, namely D2Q9 and D3Q19 illustrated in Fig.~\ref{fig:LBmodels}. The velocity vectors are enumerated as $\vect{\mathsf{e}}_i \in \R^d$, $i = 1,\dots,b$ with $b = 9$ for the D2Q9 model and $b = 19$ for the D3Q19 model. In our case, each velocity vector within a grid cell is associated with a Wigner matrix $\mathsf{W}_i(\vect{x},t) \in \mathcal{M}_2$, which can be interpreted as the average spin density at cell $\vect{x}$ with momentum $\vect{\mathsf{e}}_i$ at time $t$. For simplicity, we impose periodic boundary conditions on the lattice $\Lambda$ in the subsequent simulations.

For the BGK approximation of the collision operator in the discrete LBM model, the exact averaging in Eq.~\eqref{eq:AvrDensity}, \eqref{eq:AvrVelocity}, \eqref{eq:AvrEnergy} to obtain the moments is replaced by a multi-dimensional integration rule. In the following section~\ref{sec:Quadrature}, we will devise such an integration rule adapted to the Fermi-Dirac distribution and the discretized moments $\vect{\mathsf{e}}_i$, i.e., the points of the integration rule are precisely the $\vect{\mathsf{e}}_i$. In section~\ref{sec:WeqPoly}, we will construct a polynomial expansion of the Fermi-Dirac equilibrium distribution function $W_{\mathrm{FD}}$ compatible with the integration rule, which has the same moments as $W_{\mathrm{FD}}$. The polynomial equilibrium function is required by the discrete version of the BGK collision operator, as explained in section~\ref{sec:Algorithm}.

\subsection{Multidimensional numerical integration with Fermi-Dirac weight}
\label{sec:Quadrature}

Let us consider the following $d$-dimensional integral
\begin{equation}
\label{eq:Ih}
I[h] = \frac{1}{n_{d,\beta,\mu}} \int_{\R^d} h(\vect{p}) \left(\mathrm{e}^{\beta\left(\frac{1}{2}\vect{p}^2 - \mu\right)} + 1\right)^{-1} \ud\vect{p}
\end{equation}
for a given (smooth) function $h: \R^d \to \R$. Numerical approximations are typically of the form
\begin{equation}
\label{eq:quadratureGeneral}
I[h] \approx Q[h] = \sum_{i=1}^N w_i \, h\big(\vect{p}^{(i)}\big),
\end{equation}
with weights $w_i \in \R$ and points $\vect{p}^{(i)} \in \R^d$ independent of $h$. For $d \ge 2$, $Q[h]$ is denoted a \emph{cubature} (see \cite{Sobolev1997,CoolsCubature1997} and references therein). A cubature formula~\eqref{eq:quadratureGeneral} has degree $m$ if it is exact for polynomials $h$ of algebraic degree at most $m$, and not exact for at least one polynomial of degree $m + 1$. As usual, the degree of a monomial $\prod_{i=1}^d x_i^{\alpha_i}$ with $\alpha_i \in \N$ is given by $\abs{\vect{\alpha}} = \sum_{i=1}^d \alpha_i$, and the degree of a polynomial by the highest degree of its composing monomials. In our case, we look for fully symmetric quadrature rules \cite{SymmQuad1979} due to the radial symmetry of the Fermi-Dirac function, such that the moments \eqref{eq:Ih} are invariant under permutations and sign changes of coordinates.

As mentioned, we identify $\vect{p}^{(i)}$ with the velocity vectors $\vect{\mathsf{e}}_i$ of the D2Q9 model for $d = 2$ and the D3Q19 model for $d = 3$, as illustrated in Fig.~\ref{fig:LBmodels}.

In the classical case, the Maxwell-Boltzmann distribution separates into a product of Gaussians in each dimension, and the weights associated with each discrete velocity for the D2Q9 model can be derived via the Gauss-Hermite quadrature rule of order~$3$ \cite{HeLuoTheory1997}. However, the Fermi-Dirac distribution does not separate into a product of functions, so the following ansatz is used.

%$\vect{\mathsf{e}}_i$
For dimension $d = 2$, we enumerate the 9 velocity vectors as $\xi \cdot(i_1, i_2) \in \R^2$ with $i_1, i_2 \in \{-1,0,1\}$ and $\xi \in \R_{>0}$. The corresponding weights are labeled $w_{\vect{i}} \in \R_{>0}$. Our goal is to determine $\xi$ and $w_{\vect{i}}$ such that
\begin{equation}
\label{eq:quadrature2D}
\frac{1}{n_{2,\beta,\mu}} \int_{\R^2} h(\vect{p}) \left(\mathrm{e}^{\beta\left(\frac{1}{2}\vect{p}^2 - \mu\right)} + 1\right)^{-1} \ud\vect{p} \approx \sum_{i_1,i_2 \in \{-1,0,1\}} w_{\vect{i}} \, h(\xi\,\vect{i})
\end{equation}
for smooth functions $h: \R^2 \to \R$, with equality for polynomials up to order $5$. Due to symmetry, $w_{i_1,i_2} = w_{i_2,i_1}$ and $w_{i_1,i_2} = w_{i_1',i_2'}$ when $i_1 = \pm i_1'$ and $i_2 = \pm i_2'$, so only the weights $w_{00}, w_{01}, w_{11}$ and the length $\xi$ need to be determined. With the analytic solution
\begin{equation}
\label{eq:QuadSolution2D}
w_{00} = 1 - 5a, \quad w_{01} = a, \quad w_{11} = \frac{a}{4}, \quad a = \frac{F_1(\beta\mu)^2}{9 \, F_0(\beta\mu) F_2(\beta\mu)}, \quad \xi = \left(\frac{3}{\beta} \frac{F_2(\beta\mu)}{F_1(\beta\mu)}\right)^{1/2},
\end{equation}
the approximation~\eqref{eq:quadrature2D} becomes indeed exact for polynomials up to order $5$, i.e., for all $h(\vect{p}) = p_1^{m_1} \, p_2^{m_2}$ with $m_1 + m_2 \le 5$. Note that all odd moments vanish due to the even symmetry of the Fermi-Dirac distribution.

For $\mu = 0$, the parameters in Eq.~\eqref{eq:QuadSolution2D} simplify to
\begin{equation}
\label{eq:QuadSolution2Dmu0}
a = \frac{\pi^4}{972\, \zeta(3) \log(2)} \doteq 0.120277, \quad \xi = \frac{3}{\pi} \left(\frac{3\,\zeta(3)}{\beta}\right)^{1/2} \doteq \frac{1.8134}{\sqrt{\beta}},
\end{equation}
where $\zeta(s)$ is the Riemann zeta function.

\medskip

For dimension $d = 3$, one can proceed analogously. We enumerate the points $\xi\cdot(i_1,i_2,i_3) \in \R^3$ and weights $w_{\vect{i}} \in \R_{>0}$ of the D3Q19 model (on the right in Fig.~\ref{fig:LBmodels}) using the indices $i_1, i_2, i_3 \in \{-1,0,1\}$, with the additional constraint that $\abs{i_1} + \abs{i_2} + \abs{i_3} \le 2$ due to the missing cell corners in the D3Q19 model. The goal is an approximation
\begin{equation}
\label{eq:quadrature3D}
\frac{1}{n_{3,\beta,\mu}} \int_{\R^3} h(\vect{p}) \left(\mathrm{e}^{\beta\left(\frac{1}{2}\vect{p}^2 - \mu\right)} + 1\right)^{-1} \ud\vect{p} \approx \sum_{\vect{i} \in \mathrm{D3Q19}} w_{\vect{i}} \, h(\xi\,\vect{i}),
\end{equation}
such that equality holds for polynomials $h$ up to order $5$. Again due to symmetry, only the weights $w_{000}, w_{001}, w_{011}$ and the length $\xi$ need to be determined. A short calculation results in the analytic solution
\begin{equation}
\label{eq:QuadSolution3D}
\begin{split}
w_{000} &= 1 - 12a, \quad w_{001} = a, \quad w_{011} = \frac{a}{2}, \quad a = 
\frac{F_{3/2}(\beta\mu)^2}{18\,F_{1/2}(\beta\mu) \, F_{5/2}(\beta\mu)}; \\
\xi &= \left(\frac{3}{\beta} \frac{F_{5/2}(\beta\mu)}{F_{3/2}(\beta\mu)}\right)^{1/2}.
\end{split}
\end{equation}
Evaluated at $\mu = 0$, the parameters in Eq.~\eqref{eq:QuadSolution3D} simplify to
\begin{equation}
\label{eq:QuadSolution3Dmu0}
a = \frac{\left(41 + 9\sqrt{2}\right) \zeta(\frac{5}{2})^2}{558 \, \zeta(\frac{3}{2}) \, \zeta(\frac{7}{2})} \doteq 0.0588684, \quad \xi = \left(\frac{3}{\beta} \frac{\left(15 + 2\sqrt{2}\right) \zeta(\frac{7}{2})}{14\,\zeta(\frac{5}{2})}\right)^{1/2} \doteq \frac{1.79131}{\sqrt{\beta}}.
\end{equation}

\medskip

In the following, the inverse temperature and chemical potential of the integration rule are denoted by $\beta_0$ and $\mu_0$, respectively, to distinguish them from the physical quantities. According to Eq.~\eqref{eq:QuadSolution2D} and \eqref{eq:QuadSolution3D}, the length factor $\xi$ depends on $\beta_0$ and $\mu_0$. The uniform grid of the lattice Boltzmann methods requires the same $\xi$ for all cells, which can be identified as the lattice constant. Thus we set $\beta_0$ to a uniform, fixed value for all grid cells, and $\mu_0 = 0$. Different from that, the \emph{physical} quantities $\beta$, $\mu_{\uparrow}$ and $\mu_{\downarrow}$ of the equilibrium function associated with a grid cell can vary across cells (see the following section~\ref{sec:WeqPoly}). An alternative route (not pursued in the present paper) would be an \emph{isothermal} model, where $\beta = \beta_0$ is held fixed for all grid cells and time steps. The main disadvantage of this approach is the loss of the energy conservation in the simulation. Note that isothermal models are typically used in ``classical'' lattice Boltzmann methods \cite{HeLuoTheory1997}.

\medskip

The integration rule replaces the analytic integrals in Eq.~\eqref{eq:AvrDensity}, \eqref{eq:AvrVelocity} and \eqref{eq:AvrEnergy} for the local density, velocity and energy by the discrete analogues (at given $\vect{x}, t$)
\begin{equation}
\label{eq:DiscrAvr}
\rho = \sum_{i=1}^b \mathsf{W}_i(\vect{x},t), \quad \tr[\rho]\,\vect{u} = \sum_{i=1}^b \vect{\mathsf{e}}_i \tr[\mathsf{W}_i(\vect{x},t)], \quad \tr[\rho]\,\varepsilon = \sum_{i=1}^b \tfrac{1}{2}\abs{\vect{\mathsf{e}}_i - \vect{u}}^2 \tr[\mathsf{W}_i(\vect{x},t)].
\end{equation}
By convention, the integration weights $w_i$ are already absorbed into the discretized Wigner functions $\mathsf{W}_i(\vect{x},t)$.

\subsection{Polynomial expansion of the equilibrium distribution function}
\label{sec:WeqPoly}

Our goal is to construct an approximate equilibrium distribution function compatible with the quadrature formula Eq.~\eqref{eq:quadrature2D} at $\mu = 0$, such that its moments agree with the exact Fermi-Dirac values \eqref{eq:rhoFDstdOmega}, \eqref{eq:energyFDstdOmega} and the average velocity $\vect{u}$. The following ansatz is used:
\begin{multline}
\label{eq:WeqPoly}
W^{(\mathrm{eq})}(\vect{p}) = \frac{1}{n_{d,\beta_0,0}} \left(\mathrm{e}^{\beta_0 \frac{1}{2}\vect{p}^2} + 1\right)^{-1} U \begin{pmatrix}n_{d,\beta,\mu_{\uparrow}}& 0\\ 0& n_{d,\beta,\mu_{\downarrow}}\end{pmatrix} U^* \\
\times \Big( \alpha_1 + \alpha_2 \,\beta_0\,\tfrac{1}{2}\vect{p}^2 + \alpha_3\,\beta_0\,(\vect{p}\cdot\vect{u}) + \alpha_4 (\beta_0\, \vect{p}\cdot\vect{u})^2 + \alpha_5\,\beta_0\,\tfrac{1}{2}\vect{u}^2 \Big),
\end{multline}
with the coefficients $\alpha_i$ to be determined. Note that $\mu$ has disappeared from the exponent in Eq.~\eqref{eq:WeqPoly}, as required by the quadrature rule mentioned above. A solution for the coefficients in dimensions $d = 2$ and $d = 3$ is provided in the appendix. It turns out that $\alpha_3$, $\alpha_4$ and $\alpha_5$ are constants, and $\alpha_1$, $\alpha_2$ only depend on $\beta_0\,\varepsilon_{\mathrm{FD}}$, i.e., $\beta_0$ times the average internal energy defined in Eq.~\eqref{eq:energyFDstdOmega}. This dependence and being matrix-valued is the main difference to the conventional equilibrium function used in the classical lattice Boltzmann method; otherwise, the expression~\eqref{eq:WeqPoly} formally resembles the classical counterpart.

\begin{figure}[!ht]
\centering
\subfloat[$d = 2$, $\beta_0 = 0.8$]{
\includegraphics[width=0.4\textwidth]{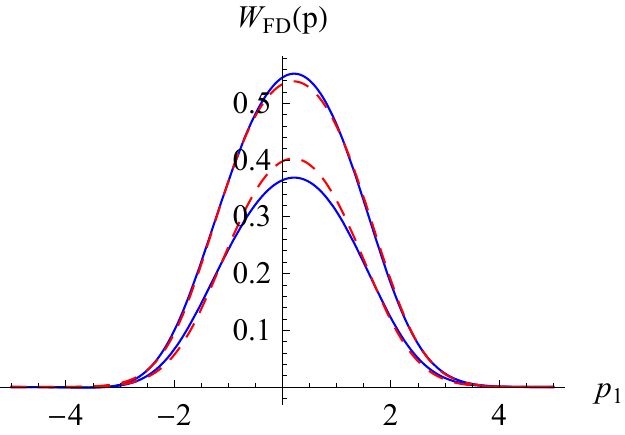}}
\hspace{0.1\textwidth}
\subfloat[$d = 3$, $\beta_0 = 0.8$]{
\includegraphics[width=0.4\textwidth]{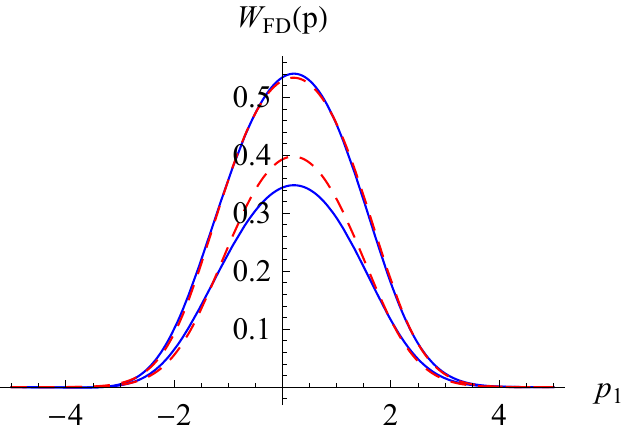}}\\
\subfloat[$d = 2$, $\beta_0 = 1.2$]{
\includegraphics[width=0.4\textwidth]{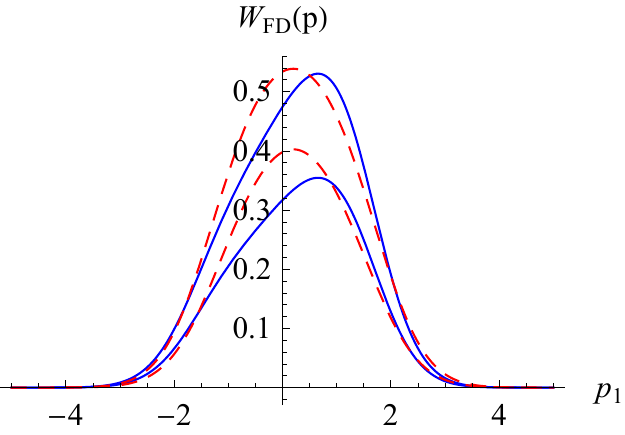}}
\hspace{0.1\textwidth}
\subfloat[$d = 3$, $\beta_0 = 1.2$]{
\includegraphics[width=0.4\textwidth]{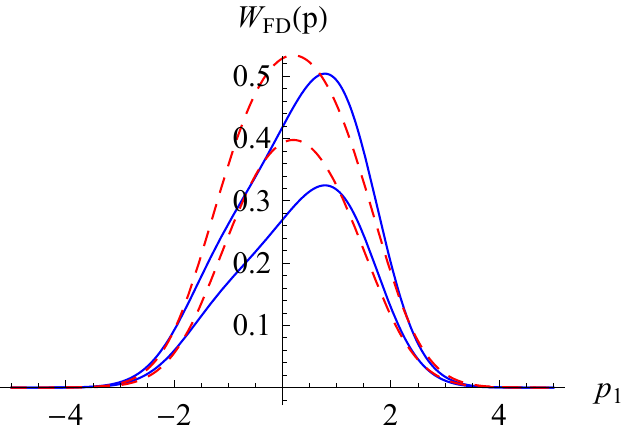}}
\caption{(Color online) Diagonal matrix entries of the polynomial approximation in Eq.~\eqref{eq:WeqPoly} with $U = \mathbbm{1}$ (blue solid curves) compared to the exact Fermi-Dirac equilibrium distribution $W_{\mathrm{FD}}$ defined in Eq.~\eqref{eq:WFermiDirac} with $\omega(\vect{p}) = \frac{1}{2} \abs{\vect{p}}^2$ (dashed red), for $d = 2$ (left column) and $d = 3$ (right column).}
\label{fig:Weq}
\end{figure}

Note that the approximation \eqref{eq:WeqPoly} of $W_{\mathrm{FD}}$ is \emph{per se} independent of the quadrature formula constructed in section~\ref{sec:Quadrature}. In particular, the moments of $W^{(\mathrm{eq})}$ can be calculated by continuous integration, according to Eq.~\eqref{eq:AvrDensity}, \eqref{eq:AvrVelocity} and \eqref{eq:AvrEnergy}. But of course the structure of $W^{(\mathrm{eq})}$ allows us to obtain the exact same moments using the quadrature formula from section~\ref{sec:Quadrature}.

Fig.~\ref{fig:Weq} compares the approximation \eqref{eq:WeqPoly} with the exact Fermi-Dirac function $W_{\mathrm{FD}}$, both for $d = 2$ and $d = 3$. The upper row differs from the lower solely by $\beta_0$, the ``reference'' inverse temperature used in the weight function of the numerical integration rule. The functions are plotted along the $p_1$ axis through the origin. The parameters defining the distribution function read $\beta = 1$, $\mu_{\uparrow} = 0.2$, $\mu_{\downarrow} = -0.35$, $\vect{u} = (0.2, 0.3)$ for $d = 2$ and $\vect{u} = (0.2, 0.3, -0.2)$ for $d = 3$. One notices the larger deviation from the exact curve as $\beta_0$ overshoots the physical $\beta$, a feature which seems to be systematic for the polynomial approximation.

In terms of the velocity vectors $\vect{\mathsf{e}}_i$, the equilibrium function \eqref{eq:WeqPoly} evaluated at these points times the corresponding weight $w_i$ gives the discretized equilibrium distribution function
\begin{multline}
\label{eq:WeqDiscr}
\mathsf{W}^{(\mathrm{eq})}_i = w_i \, U \begin{pmatrix}n_{d,\beta,\mu_{\uparrow}}& 0\\ 0& n_{d,\beta,\mu_{\downarrow}}\end{pmatrix} U^* \\
\times \Big( \alpha_1 + \alpha_2\,\beta_0\,\tfrac{1}{2}\vect{\mathsf{e}}_i^2 + \alpha_3 \,\beta_0\,(\vect{\mathsf{e}}_i\cdot\vect{u}) + \alpha_4 (\beta_0\,\vect{\mathsf{e}}_i\cdot\vect{u})^2 + \alpha_5\,\beta_0\,\tfrac{1}{2}\vect{u}^2 \Big),
\end{multline}
which will be used in the subsequent lattice Boltzmann simulation.

The approximate polynomial equilibrium function constructed so far has the drawback that the last factor in Eq.~\eqref{eq:WeqDiscr} can become negative for certain values of $\beta$, $\mu_{\uparrow}$, $\mu_{\downarrow}$ and $\vect{u}$, resulting in an (unphysical) $\mathsf{W}^{(\mathrm{eq})}_i$ with negative eigenvalues. This issue is discussed in the following section.

\subsection{Quantum lattice Boltzmann method}
\label{sec:Algorithm}

Using the results from Sec.~\ref{sec:Quadrature} and \ref{sec:WeqPoly}, we can now specify the details of the numerical framework using the BGK collision operator. The framework takes as input the relaxation time $\tau$ and a reference inverse temperature $\beta_0$ (globally constant) for the numerical integration (see section~\ref{sec:Quadrature}). The weights $w_i$ and length parameter $\xi$ can then be precomputed according to Eq.~\eqref{eq:QuadSolution2Dmu0} for $d = 2$ and Eq.~\eqref{eq:QuadSolution3Dmu0} for $d = 3$. The transport step in Eq.~\eqref{eq:TransportStep} below relates $\xi$ to the time step $\Delta t$ and lattice cell width $\Delta x$ via $\Delta t\, \xi = \Delta x$. Furthermore, an initial configuration of distribution functions $\mathsf{W}_i(\vect{x},t_0)$ is required as input.

The framework approximates the kinetics of the Boltzmann equation by performing the following calculations for each discrete time step $t \to t + \Delta t$:
\begin{itemize}

\item For each cell $\vect{x}$, the local average momentum $\rho$, velocity  $\vect{u}$ and energy $\varepsilon$ of $\mathsf{W}_i(\vect{x}, t)$ are calculated according to Eq.~\eqref{eq:DiscrAvr}. These averages define the local equilibrium function
\begin{equation}
\label{eq:WeqAlgo}
\mathsf{W}^{(\mathrm{eq})}_i = w_i \, \rho\, \Big( \alpha_1 + \alpha_2\,\beta_0\,\tfrac{1}{2}\vect{\mathsf{e}}_i^2 + \alpha_3 \,\beta_0\,(\vect{\mathsf{e}}_i\cdot\vect{u}) + \alpha_4 (\beta_0\,\vect{\mathsf{e}}_i\cdot\vect{u})^2 + \alpha_5\,\beta_0\,\tfrac{1}{2}\vect{u}^2 \Big),
\end{equation}
with the coefficients $\alpha_j$ provided in the appendix (setting $\varepsilon_{\mathrm{FD}} = \varepsilon$). The unitary matrix $U$ appearing in Eq.~\eqref{eq:WeqDiscr} does not need to be computed, nor any Fermi-Dirac integral function $F_k(x)$.

\item Collision: for each cell $\vect{x}$, the discretized BGK collision operator is applied to obtain the post-collisional distribution function $\mathsf{W}^{\mathrm{coll}}_i(\vect{x},t)$:
\begin{equation}
\label{eq:CollisionStep}
\mathsf{W}^{\mathrm{coll}}_i(\vect{x},t) = \mathsf{W}_i(\vect{x},t) + \frac{\Delta t}{\tau} \left(\mathsf{W}^{\mathrm{(eq)}}_i - \mathsf{W}_i(\vect{x},t)\right),
\end{equation}
with $\mathsf{W}^{\mathrm{(eq)}}_i$ defined in Eq.~\eqref{eq:WeqAlgo}.

\item Streaming: the distribution functions are transported along their respective velocity directions $\vect{\mathsf{e}}_i$:
\begin{equation}
\label{eq:TransportStep}
\mathsf{W}_i(\vect{x}+\Delta t\,\vect{\mathsf{e}}_i, t + \Delta t) = \mathsf{W}^{\mathrm{coll}}_i(\vect{x},t)
\end{equation}
for all $i = 1,\dots,b$ and cells $\vect{x}$.
This step approximates the transport term \\ $\vect{p} \cdot \vect{\nabla}_{\vect{x}} W(\vect{p},\vect{x},t)$ in the Boltzmann equation~\eqref{eq:BoltzmannEquation}.
\end{itemize}
The numerical operations required by the algorithm are relatively simple, and the Fermi-Dirac integral functions $F_k(x)$ need not be evaluated.

To compensate for potentially negative-definite, unphysical $\mathsf{W}^{(\mathrm{eq})}_i$ in \eqref{eq:WeqAlgo}, one could dynamically decrease $\frac{1}{\tau}$ in \eqref{eq:CollisionStep} such that $\mathsf{W}^{\mathrm{coll}}_i(\vect{x},t)$ is guaranteed to be positive semidefinite. However, this could lead to $\frac{1}{\tau} \to 0$, such that the simulation effectively halts. We have therefore avoided this adjustment; for the numerical examples in Sec.~\ref{sec:Simulation}, negative-definite distribution functions appear rarely and do not seem to affect the simulation results.

Importantly, the numerical scheme preserves the local, discrete versions of the density, momentum and energy conservation laws. For periodic boundary conditions without external forces, it follows that the global average density, velocity and total energy
\begin{align}
\label{eq:GlobalDensity}
\rho(t) &= \frac{1}{\abs{\Lambda}} \sum_{\vect{x} \in \Lambda} \rho(\vect{x},t)
= \frac{1}{\abs{\Lambda}} \sum_{\vect{x} \in \Lambda} \sum_{i=1}^b \mathsf{W}_i(\vect{x},t),\\
\label{eq:GlobalVelocity}
\tr[\rho(t)]\, \vect{u}(t) &= \frac{1}{\abs{\Lambda}} \sum_{\vect{x} \in \Lambda} \tr[\rho(\vect{x},t)]\,\vect{u}(\vect{x},t) = \frac{1}{\abs{\Lambda}} \sum_{\vect{x} \in \Lambda} \sum_{i=1}^b \vect{\mathsf{e}}_i \tr[\mathsf{W}_i(\vect{x},t)], \quad \\
\begin{split}
\label{eq:GlobalTotalEnergy}
\tr[\rho(t)]\, \varepsilon_{\mathrm{tot}}(t) &= \frac{1}{\abs{\Lambda}} \sum_{\vect{x} \in \Lambda} \tr[\rho(\vect{x},t)]\left(\varepsilon(\vect{x},t) + \tfrac{1}{2} \abs{\vect{u}(\vect{x},t)}^2\right) \\
&= \frac{1}{\abs{\Lambda}} \sum_{\vect{x} \in \Lambda} \sum_{i=1}^b \tfrac{1}{2} \abs{\vect{\mathsf{e}}_i}^2 \tr[\mathsf{W}_i(\vect{x},t)]
\end{split}
\end{align}
remain constant under the numerical time evolution. Since the final ($t \to \infty$) state is expected to be homogeneous (uniform in $\vect{x}$), the the discussion following the proof of proposition~\ref{prop:MapMomentsParams} applies, and the average quantities \eqref{eq:GlobalDensity}, \eqref{eq:GlobalVelocity}, \eqref{eq:GlobalTotalEnergy} actually determine the final equilibrium state.

Notably, the analysis in \cite{Wagner1998} suggests that a discrete version of the H-theorem does not exist for the present framework.

\section{Simulation and validation}
\label{sec:Simulation}

\paragraph{Riemann problem} To validate the quantum LBM scheme, we compare it to an analytic solution of the Riemann problem \cite{LeVeque2002} for the Euler equations \eqref{eq:Euler}, which are expected to be a good approximation close to thermal equilibrium.
\begin{figure}[!b]
\centering
\subfloat[density ${\tr[\rho]}$]{
\includegraphics[width=0.3\textwidth]{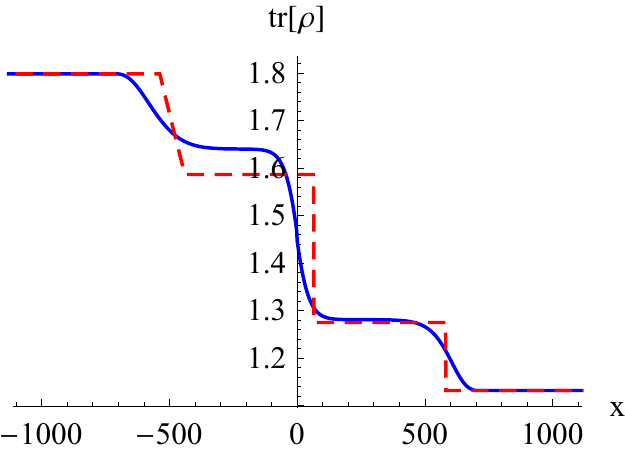}}
\hspace{0.02\textwidth}
\subfloat[velocity]{
\includegraphics[width=0.3\textwidth]{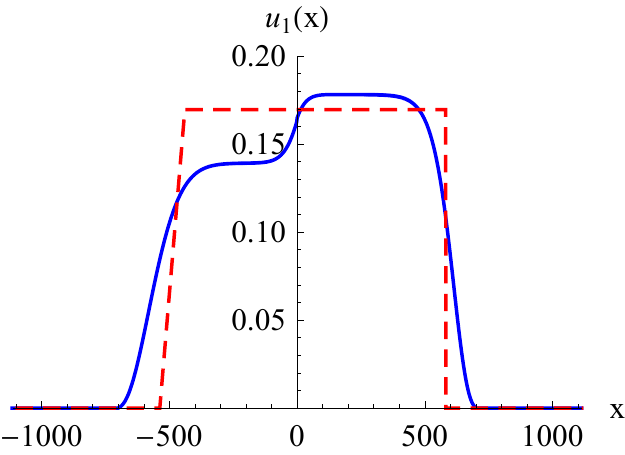}}
\hspace{0.02\textwidth}
\subfloat[pressure ${P = (\gamma-1) \tr[\rho] \varepsilon}$]{
\includegraphics[width=0.3\textwidth]{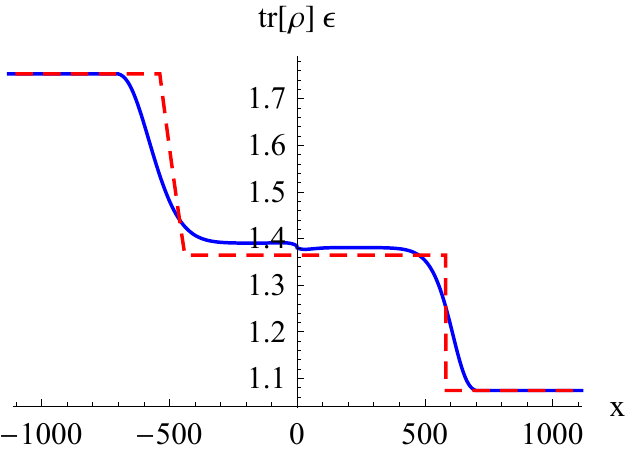}}
\caption{(Color online) Solution of a Riemann problem using the present LBM scheme after 385 time steps (blue solid curves) in comparison with a corresponding analytic solution of the Euler equations \eqref{eq:Euler} (red dashed).}
\label{fig:EulerRiemann}
\end{figure}
Specifically, the initial state consists of two thermodynamically equilibrated regions which are divided at $x = 0$ by a virtual membrane. At $t = 0$ the membrane is removed and the two regions start to intermingle with each other. In the left region we initially choose $W_{\mathrm{FD}}$ with $\mu_{\uparrow} = -\frac{3}{2}$ and $\mu_{\downarrow} = -\frac{5}{2}$, and in the right region $\mu_{\uparrow} = -2$ and $\mu_{\downarrow} = -3$. The initial inverse temperature $\beta = 1$ and velocity $\vect{u} = 0$ are the same for both regions. The corresponding density and energy can be computed according to \eqref{eq:rhoFDstdOmega} and \eqref{eq:energyFDstdOmega}. The analytical solution \cite{LeVeque2002} of the Euler equations depends only on the ratio $\frac{x}{t}$. With the described initial data, it consists of a shock wave traveling to the right, a contact discontinuity with a jump in density traveling also to the right at lower speed, and a rarefaction wave traveling to the left. For an ideal polytropic gas, the pressure is given by $P = (\gamma - 1) \tr[\rho] \varepsilon$ with $\gamma$ the adiabatic exponent. In our case $\gamma = 1 + \frac{2}{d} = 2$ since we run the quantum LBM scheme on a two-dimensional grid with dimensions $2048 \times 2$ (quasi-1D) and periodic boundary conditions. The relaxation time of the collision operator was set to $\tau = 10$. The LBM solution is compared to the analytic Euler solution in Fig.~\ref{fig:EulerRiemann}. The features of the Euler solution are qualitatively reproduced, albeit with a small jump in velocity close to $x = 0$.

\paragraph{Interference} Next, we study the interference pattern emerging from the collision of two wavelet-shaped density distribution ``packets'' in two dimensions, as shown in Fig.~\ref{fig:SpinwaveDensity}.
\begin{figure}[!b]
\centering
\subfloat[{$\tr[\rho]$} at $t = 0$]{
\includegraphics[width=0.3\textwidth]{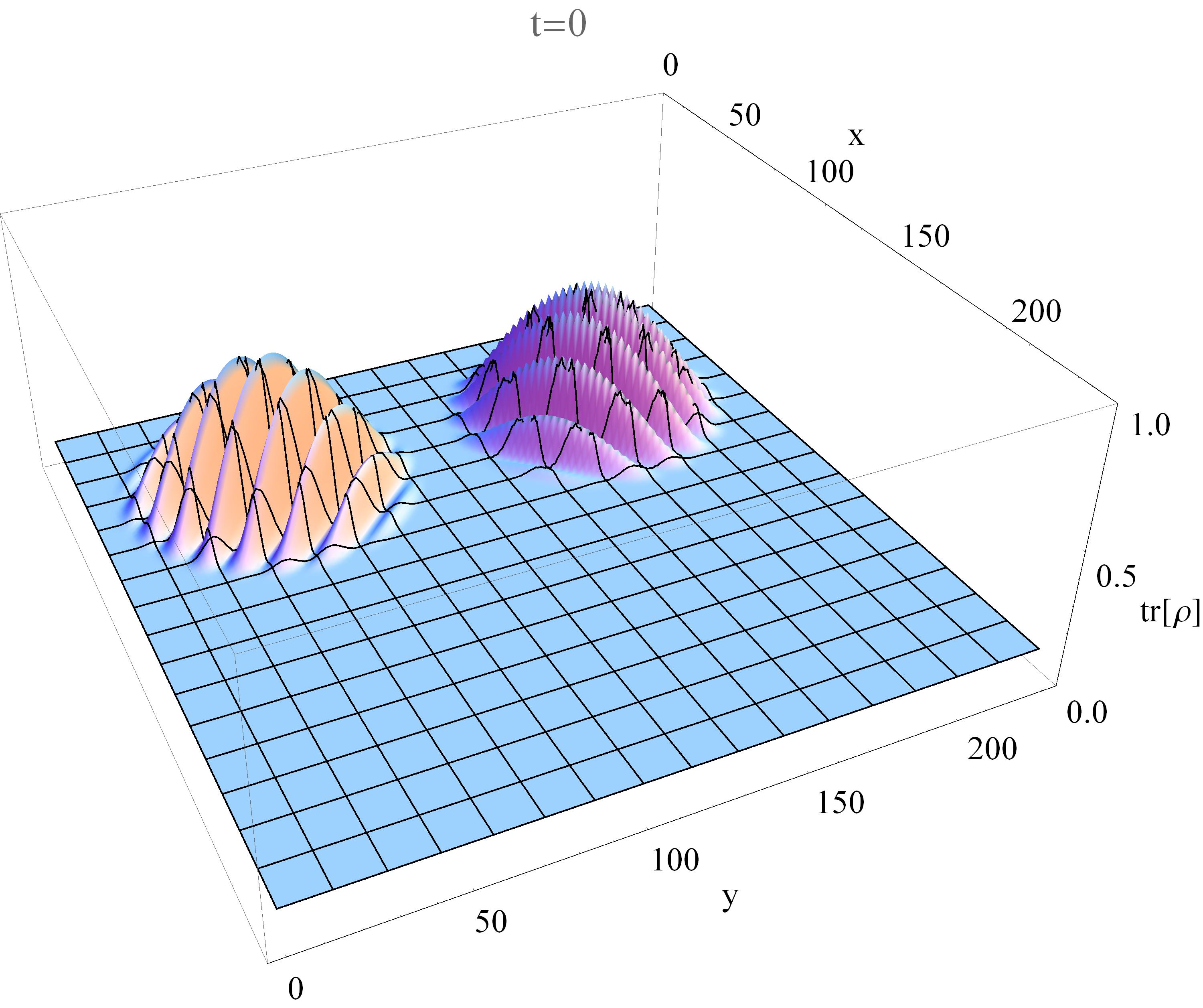}}
\hspace{0.02\textwidth}
\subfloat[{$\tr[\rho]$} at time step $32$]{
\includegraphics[width=0.3\textwidth]{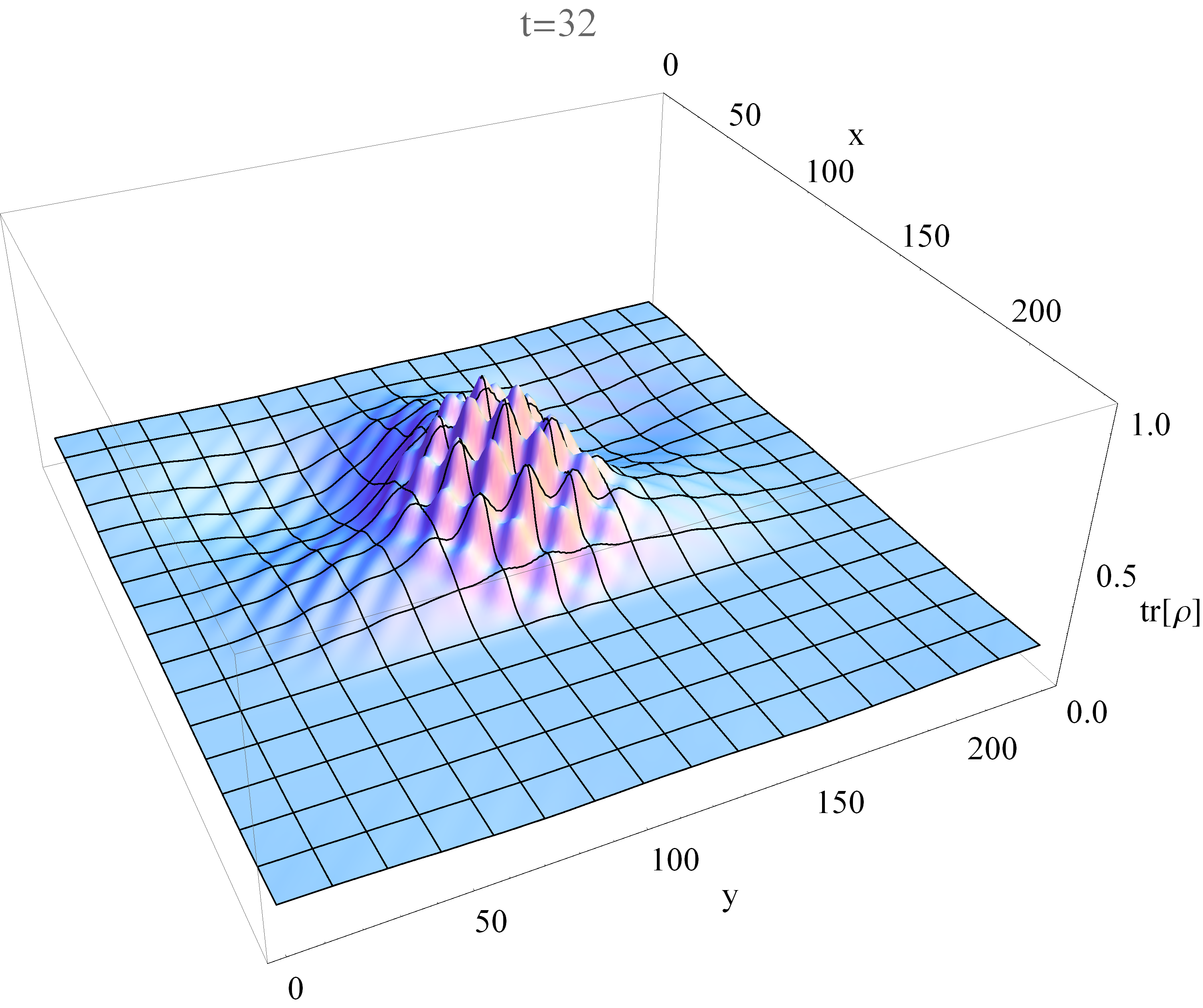}}
\hspace{0.02\textwidth}
\subfloat[{$\tr[\rho]$} at time step $54$]{
\includegraphics[width=0.3\textwidth]{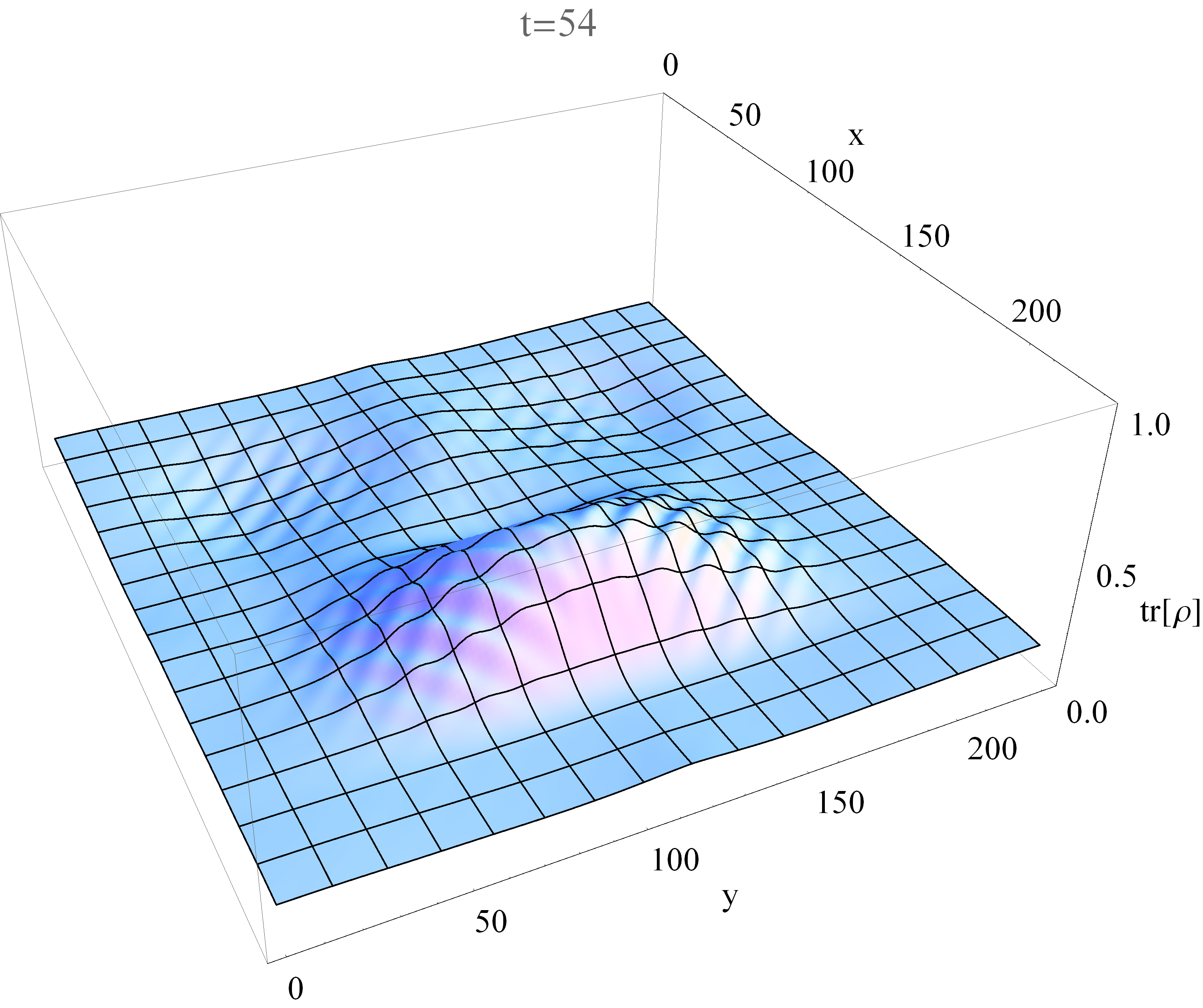}}\\
\subfloat[Bloch vectors at $t = 0$]{
\includegraphics[width=0.3\textwidth]{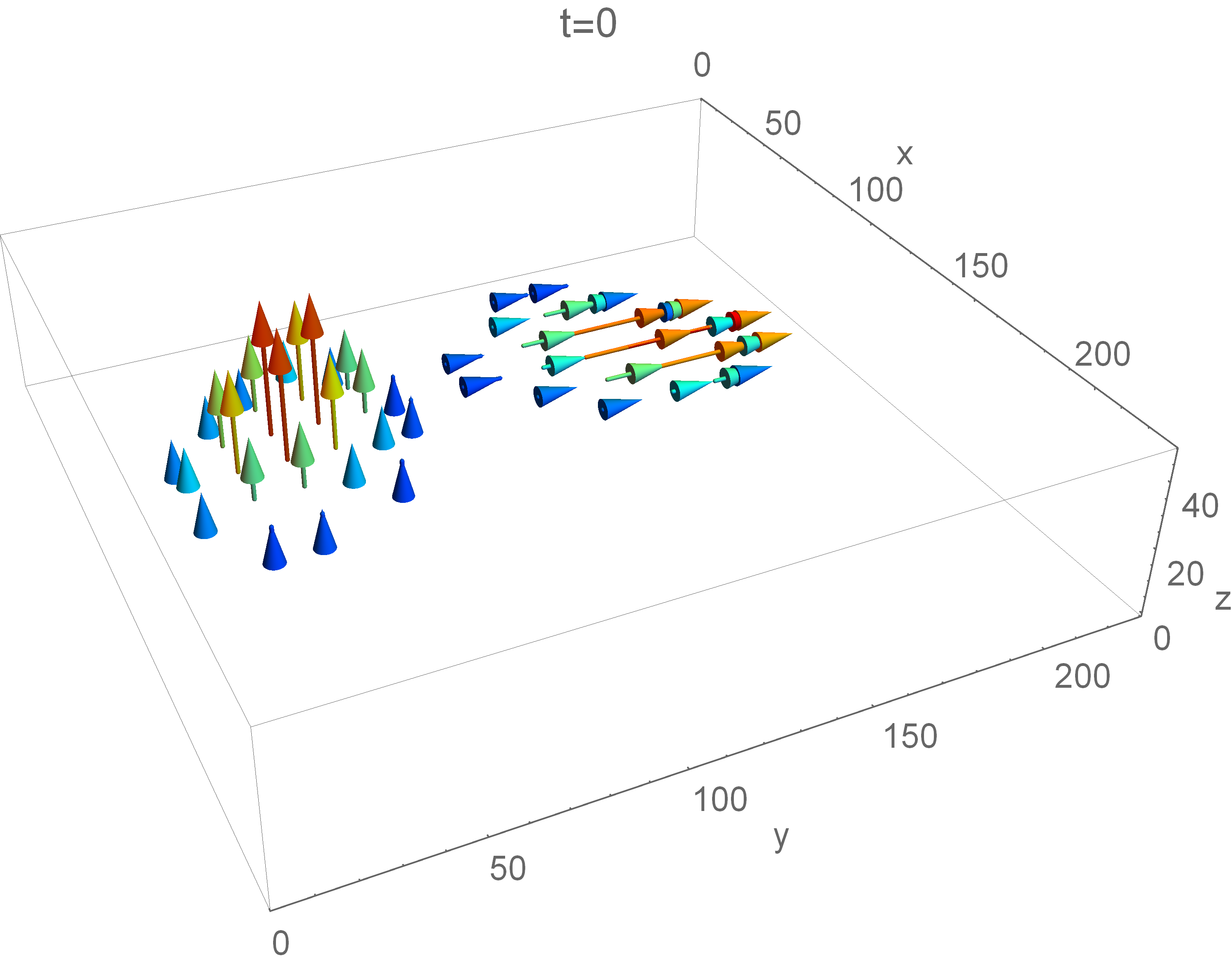}}
\hspace{0.02\textwidth}
\subfloat[Bloch vectors at time step $32$]{
\includegraphics[width=0.3\textwidth]{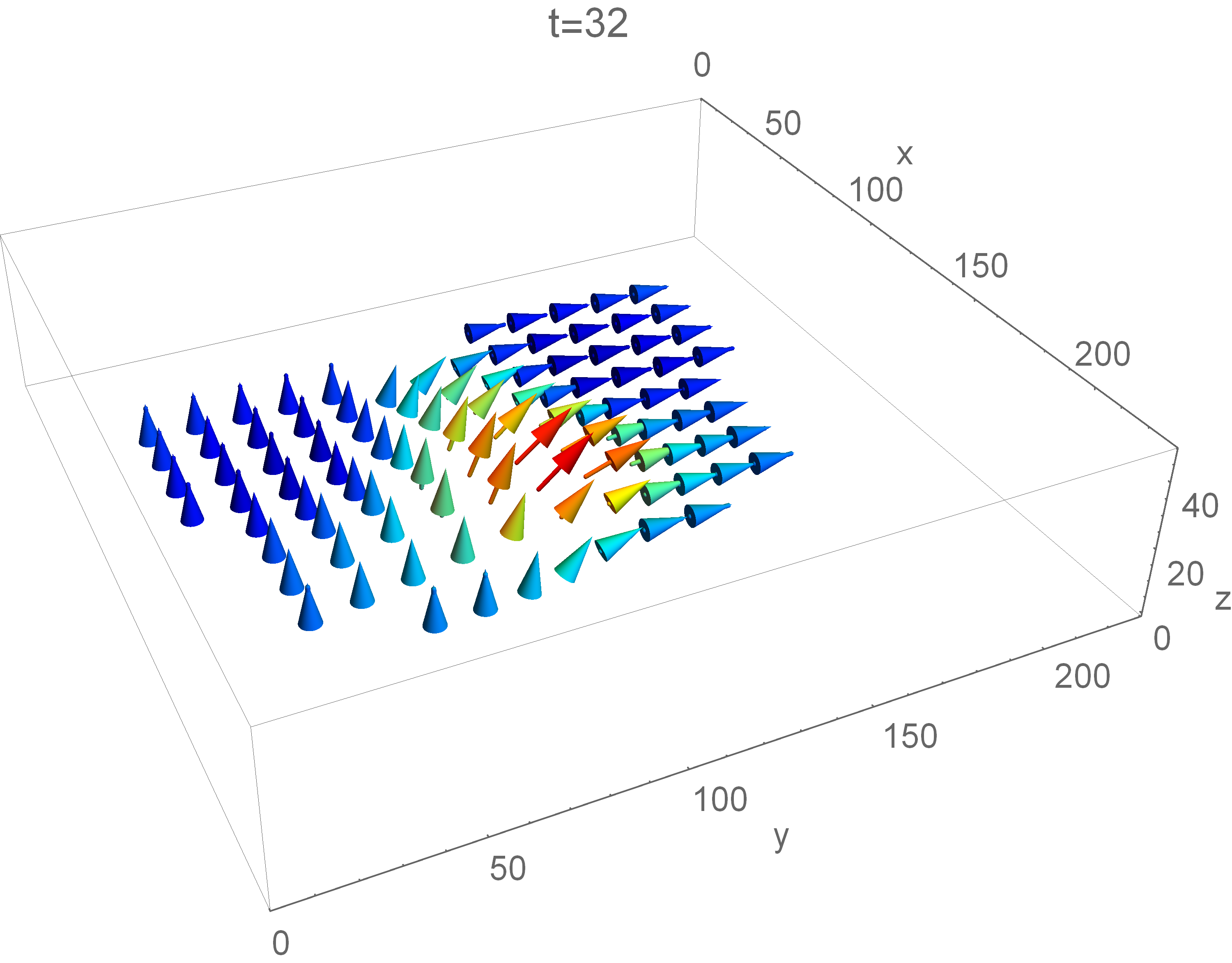}}
\hspace{0.02\textwidth}
\subfloat[Bloch vectors at time step $54$]{
\includegraphics[width=0.3\textwidth]{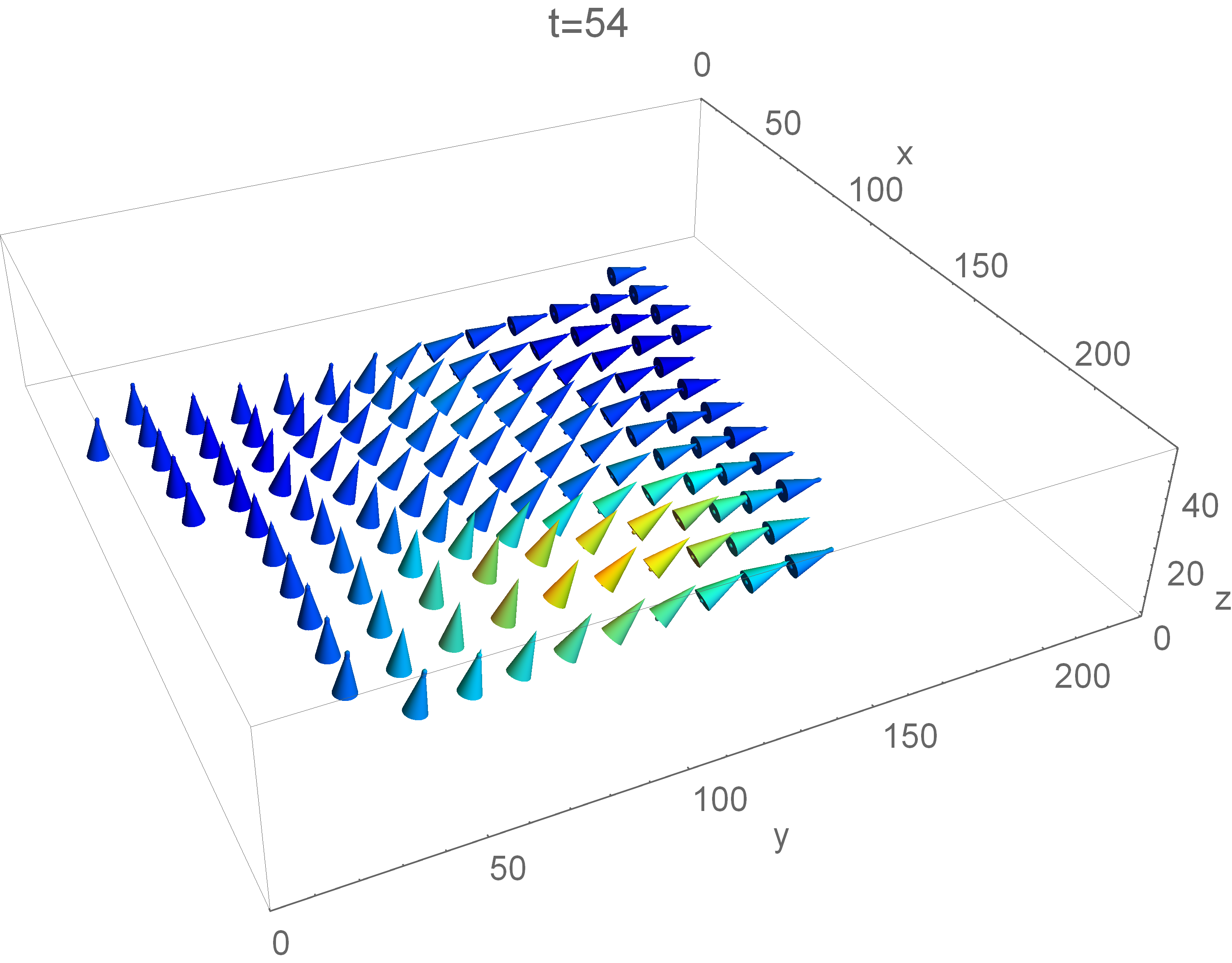}}
\caption{(Color online) Collision and interference of two wavelet-shaped density distributions with different spin orientations (Bloch vectors). The top row shows the trace of the spin density matrix $\rho(\vect{x},t)$, and the bottom row the corresponding Bloch vectors (with color encoding the density). The Bloch vectors have been scaled by the factor 400 for visual clarity.}
\label{fig:SpinwaveDensity}
\end{figure}
The top row illustrates the trace of the density matrix $\rho(\vect{x},t)$ at various time points, and the bottom row the corresponding Bloch vectors $\vect{r} \in \R^3$ (for each $\vect{x}, t$) with components $r_i = \tr[ \rho\, \sigma_i]$, where $\sigma_i$ are the Pauli spin matrices. The density matrix is uniquely determined by its trace and Bloch vector. The simulation domain consists of $128 \times 128$ grid cells, the relaxation time of the collision operator was set to $\tau = 10$, and the reference inverse temperature to $\beta_0 = 1$. The corresponding lattice constant is $\xi \doteq 1.8134$ according to Eq.~\eqref{eq:QuadSolution2Dmu0}. Initially at $t = 0$, the two packets have different spin orientations: the Bloch vectors of the left packet point in positive $z$-direction, while the Bloch vectors of the other packet point into $y$-direction. The velocities of the two packets at $t = 0$ have directions $\frac{1}{\sqrt{2}} (1, 1)$ and $\frac{1}{\sqrt{2}} (1, -1)$, respectively, such that the packets eventually collide and interfere, as illustrated at $t = 32$ in Fig.~\ref{fig:SpinwaveDensity}. After the collision, there is a smooth central wave traveling into direction $(1, 0)$, shown in the right column of Fig.~\ref{fig:SpinwaveDensity}. One notices that the Bloch vectors form a smooth transition from the $z$-direction to the $y$-direction after the collision.

\begin{figure}[!ht]
\centering
\includegraphics[width=0.5\textwidth]{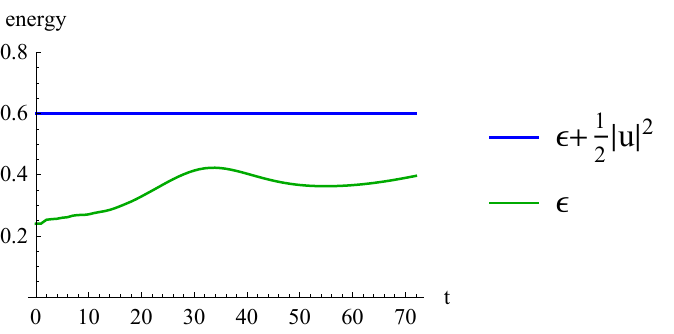}
\caption{(Color online) Time evolution of the global energy corresponding to Fig.~\ref{fig:SpinwaveDensity}. The total energy (internal energy plus kinetic energy, blue) defined in Eq.~\eqref{eq:GlobalTotalEnergy} remains constant, as expected. The global internal energy (green) can vary in time.}
\label{fig:SpinwaveEnergy}
\end{figure}
The global energy corresponding to the simulation in Fig.~\ref{fig:SpinwaveDensity} is visualized in Fig.~\ref{fig:SpinwaveEnergy}. The total energy defined in Eq.~\eqref{eq:GlobalTotalEnergy} remains indeed constant up to numerical precision, while the internal energy varies in time.

\begin{figure}[!ht]
\centering
\subfloat[Bloch vectors at $t = 0$]{
\includegraphics[width=0.3\textwidth]{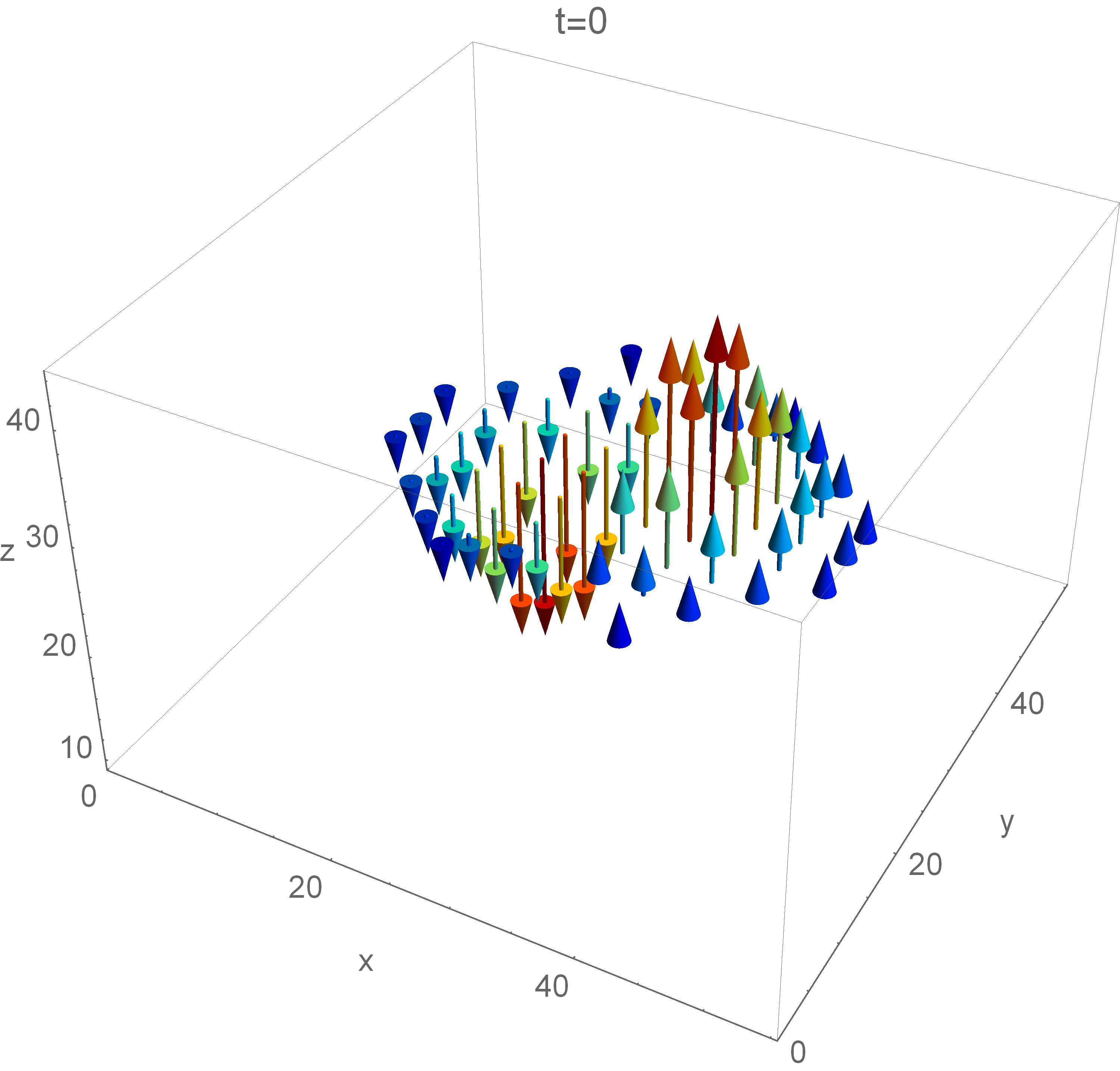}
\label{fig:Bloch3DInitial}}
\hspace{0.02\textwidth}
\subfloat[outline of $r_z$ at $t = 0$]{
\includegraphics[width=0.3\textwidth]{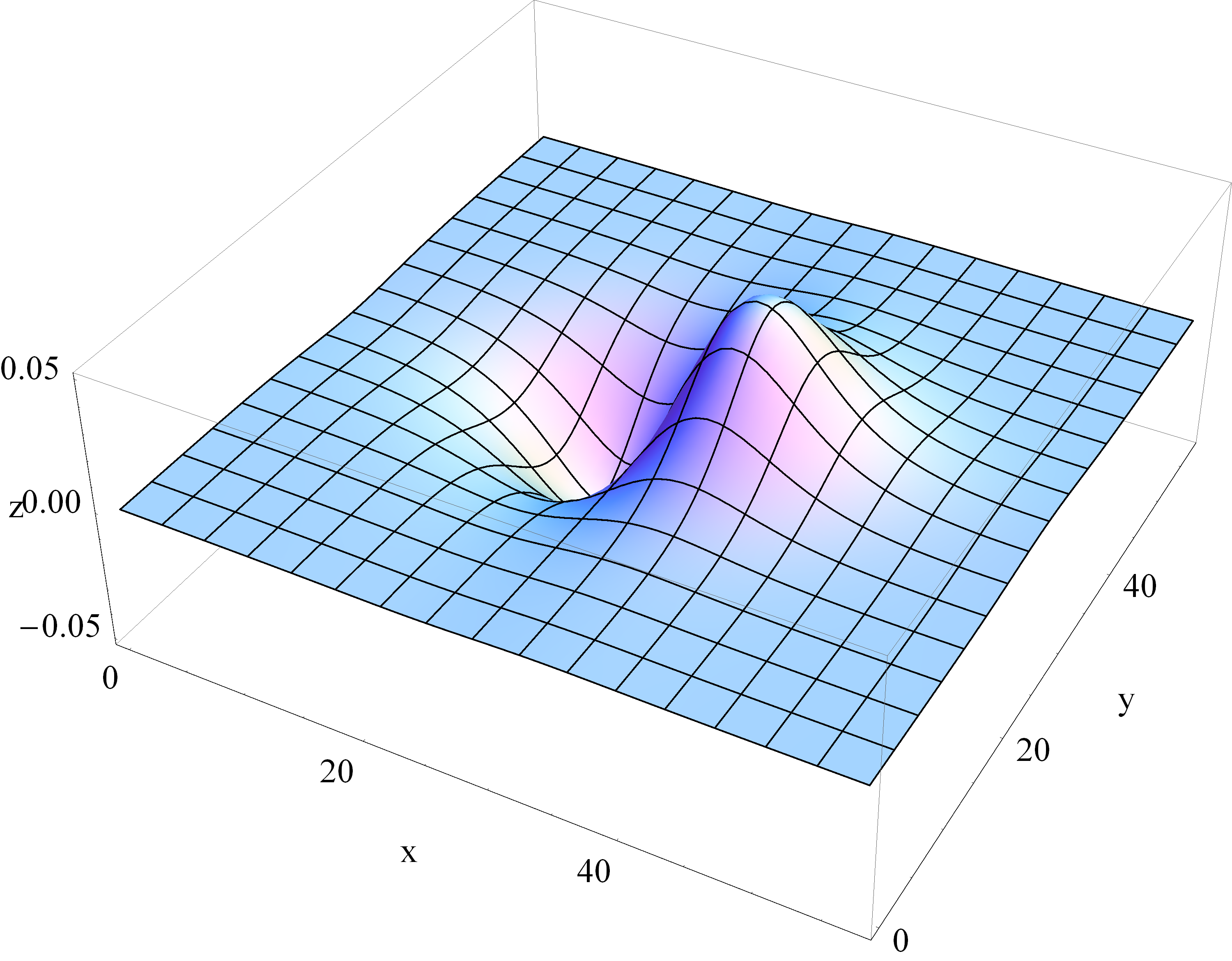}
\label{fig:Bloch3DOutline}}
\hspace{0.02\textwidth}
\subfloat[Bloch vectors at $t = 4$]{
\includegraphics[width=0.3\textwidth]{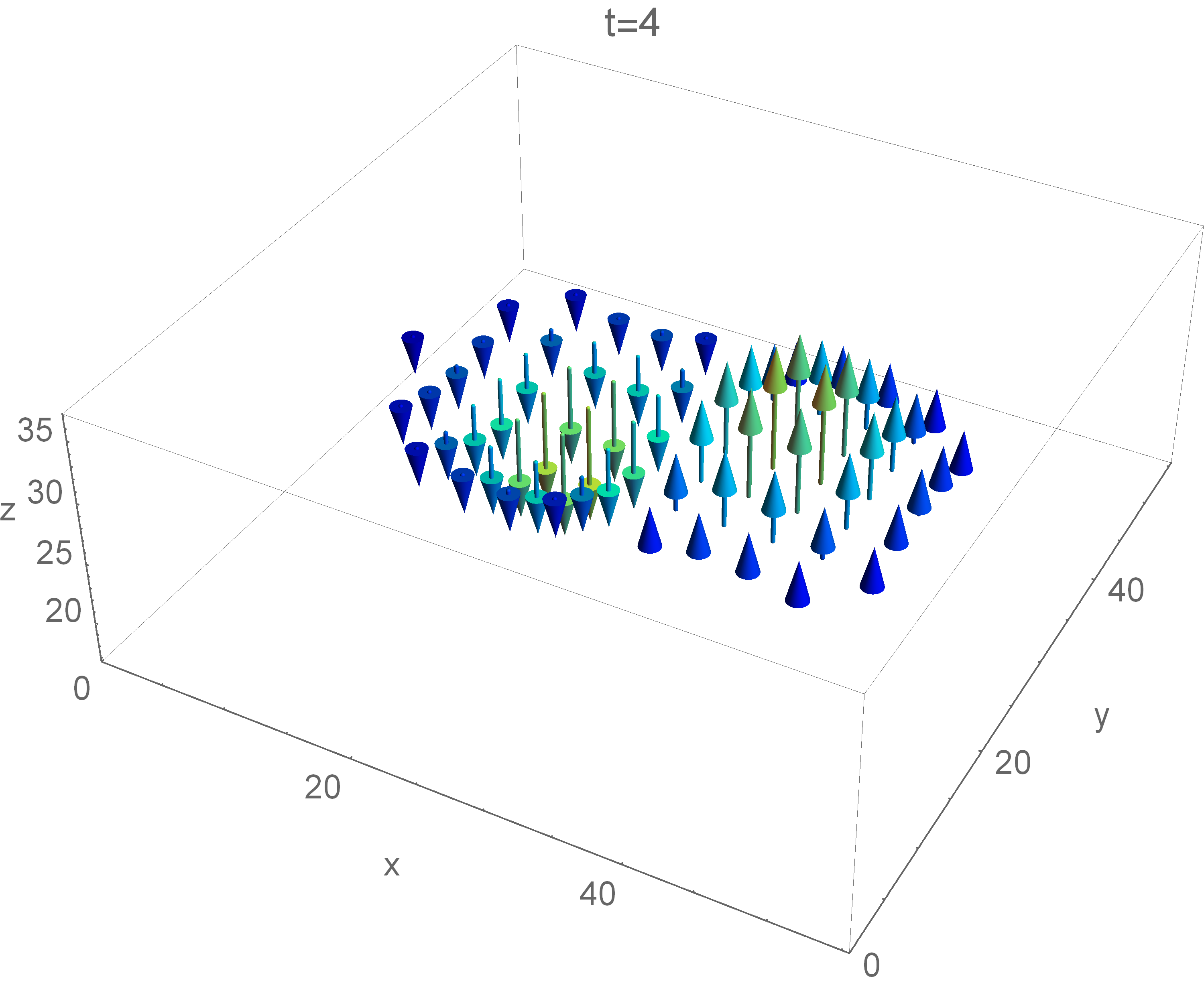}}\\
\subfloat[density matrix convergence]{
\includegraphics[width=0.3\textwidth]{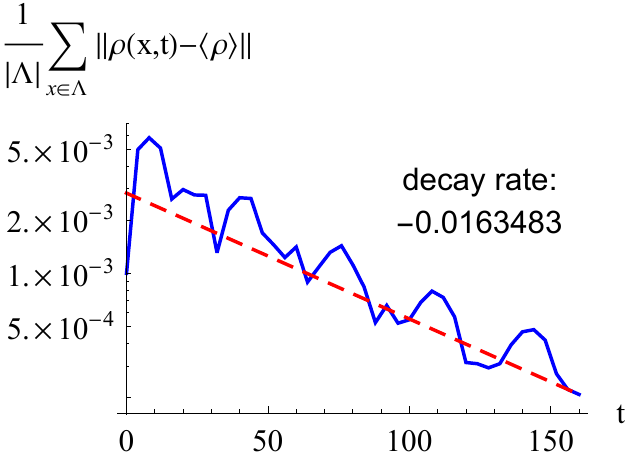}}
\hspace{0.02\textwidth}
\subfloat[velocity convergence]{
\includegraphics[width=0.3\textwidth]{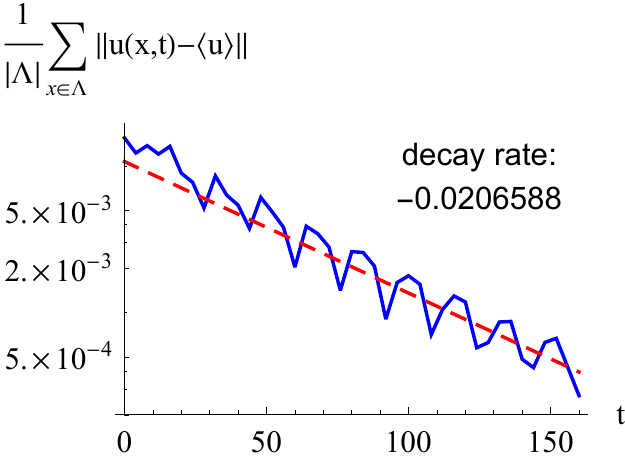}}
\hspace{0.02\textwidth}
\subfloat[total energy convergence]{
\includegraphics[width=0.3\textwidth]{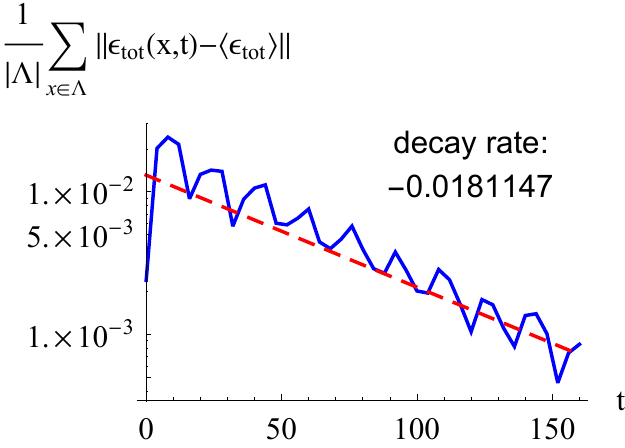}}
\caption{(Color online) Exponential convergence to equilibrium in a 3D simulation. The Bloch vectors of the initial state point along the $z$-axis. (a) shows the vectors within the central $z$-plane, after rescaling them by the factor $350$ for visual clarity, and (b) the arrow tips as smooth outline. The Bloch vectors tend towards the mean zero vector very rapidly, as illustrated in (c) after $4$ time steps only. Color in (a) and (c) encodes the length of the vectors. Bottom row: Quantification of the exponential convergence of the density matrices, velocity and total energy towards their uniform equilibrium values.}
\label{fig:Bloch3D}
\end{figure}

\paragraph{Convergence to equilibrium} The next example in Fig.~\ref{fig:Bloch3D} illustrates the convergence of the density matrices $\rho(\vect{x},t)$, velocity $\vect{u}(\vect{x},t)$ and total energy $\varepsilon_{\mathrm{tot}}(\vect{x},t)$ to the constant mean values defined in Eqs.~\eqref{eq:GlobalDensity} -- \eqref{eq:GlobalTotalEnergy}. The simulation domain consists of $32 \times 32 \times 32$ cells, the relaxation constant was set to $\tau = 10$, and the reference inverse temperature to $\beta_0 = 1$. The trace of the density matrix is chosen to be uniform, but the $z$-components of the initial Bloch vectors have a dichotomic outline, as shown in Fig.~\ref{fig:Bloch3DInitial} and \ref{fig:Bloch3DOutline}. Physically, this could correspond to neighboring regions within a solid with opposing (average) electronic spins. The initial velocity points in $x$-direction, with values proportional to the $z$-components of the Bloch vectors. As $t \to \infty$, one expects that the density, velocity and energy become uniform across the simulation domain, and thus converge to their global (equilibrium) average values. Indeed exponential convergence is confirmed numerically in the bottom row of Fig.~\ref{fig:Bloch3D}. The oscillatory pattern is likely due to waves which travel in opposite directions and repeatedly collide due to the periodic boundary conditions.

\section{Conclusions and outlook}

We have developed and implemented a lattice Boltzmann scheme complying with quantum aspects by using Fermi-Dirac equilibrium functions and taking the electronic spin explicitly into account. Still open is a quantification of the error introduced by the BGK approximation as compared to the physical collision operator \cite{BoltzmannHubbard2012}, as well as a calibration of the relaxation time $\tau$. There are also many desirable features left for future work: in regard of physical applications, an external magnetic field should be included, and various boundary conditions (other than the periodic conditions used in the present study) and additional cell types (like ``obstacle'' cells appearing in classical LBM) could be incorporated. Concerning the first point, a magnetic field adds a local Vlasov-type commutator
\begin{equation}
\label{eq:sigmaB_comm}
-i \left[ \mu \,\vect{\sigma} \cdot \vect{B}(\vect{x},t), W(\vect{p},\vect{x},t) \right]
\end{equation}
to the right side of Eq.~\eqref{eq:BoltzmannEquation}, where $\mu$ is the magnetic moment, $\vect{\sigma} = (\sigma_x, \sigma_y, \sigma_z)$ the Pauli matrices and $\vect{B}(\vect{x},t)$ the magnetic field. The term \eqref{eq:sigmaB_comm} can simply be added to the collision operator in the implementation.

%Giant magnetoresistance (GMR) in thin, alternating ferromagnetic and non-ferromagnetic layers: very weak magnetic changes lead to large differences of the electrical resistance; layered magnetic structures

\paragraph{Acknowledgments}: I'd like to thank Martin F\"urst, Hans Hiptmair and Herbert Spohn for many helpful discussions.

\newpage

\appendix

\section{Coefficients of the polynomial equilibrium functions}

For $d = 2$, a solution for the coefficients in Eq.~\eqref{eq:WeqPoly} reads
\begin{align*}
\alpha_1 &= \alpha_{\text{fac}} \left(3\,\zeta(3) - \frac{\pi^2}{6} \beta_0\,\varepsilon_{\mathrm{FD}}\right),\\
\alpha_2 &= \alpha_{\text{fac}} \left(2 \log(2) \beta_0\,\varepsilon_{\mathrm{FD}} - \frac{\pi^2}{6}\right),\\
\alpha_3 &= \frac{12 \log(2)}{\pi^2}, \\
\alpha_4 &= \alpha_{\text{fac}} \log(2), \quad \alpha_5 = -\alpha_{\text{fac}} \frac{\pi^2}{6},\\
\alpha_{\text{fac}} &= \frac{72 \log(2)}{216\log(2)\,\zeta(3) - \pi^4} \doteq 0.60447
\end{align*}
with $\varepsilon_{\mathrm{FD}}$ defined in Eq.~\eqref{eq:energyFDstdOmega} for $d = 2$ and $\zeta(s)$ the Riemann zeta function.

\medskip

Similarly, for dimension $d = 3$ one arrives at the solution
\begin{align*}
\alpha_1 &= \alpha_{\text{fac}} \left(\frac{5}{4} \left(9 - 5\sqrt{2}\right) \zeta\!\left(\tfrac{7}{2}\right) - \left(5 - 3\sqrt{2}\right) \zeta\!\left(\tfrac{5}{2}\right) \beta_0\,\varepsilon_{\mathrm{FD}}\right),\\
\alpha_2 &= \alpha_{\text{fac}} \left( \left(6 - 4 \sqrt{2}\right) \zeta\!\left(\tfrac{3}{2}\right) \tfrac{2}{3} \beta_0\,\varepsilon_{\mathrm{FD}} - \left(5 - 3\sqrt{2}\right) \zeta\!\left(\tfrac{5}{2}\right)\right), \\
\alpha_3 &= \frac{2\,\zeta\!\left(\frac{3}{2}\right)}{\left(3+\sqrt{2}\right) \zeta\!\left(\frac{5}{2}\right)}, \\
\alpha_4 &= \alpha_{\text{fac}} \left(3 - 2\sqrt{2}\right) \zeta\!\left(\tfrac{3}{2}\right), \quad
\alpha_5 = -\alpha_{\text{fac}} \left(5 - 3\sqrt{2}\right) \zeta\!\left(\tfrac{5}{2}\right),\\
\alpha_{\text{fac}} &= \frac{4\,\zeta\!\left(\frac{3}{2}\right)}{5 \left(9 - 5\sqrt{2}\right) \zeta\!\left(\frac{3}{2}\right) \zeta\!\left(\frac{7}{2}\right) - 3 \left(9 - 4\sqrt{2}\right) \zeta\!\left(\frac{5}{2}\right)^2} \doteq 1.01062
\end{align*}
with $\varepsilon_{\mathrm{FD}}$ defined in Eq.~\eqref{eq:energyFDstdOmega} for $d = 3$.

\newpage

%\bibliographystyle{unsrtmod}
%\bibliography{references}

\end{document}